\documentclass{lmcs}
\pdfoutput=1

\usepackage[utf8]{inputenc}
\usepackage{lastpage}
\lmcsdoi{17}{4}{14}
\lmcsheading{}{\pageref{LastPage}}{}{}%
{Jan.~08,~2020}{Nov.~29,~2021}{}

\usepackage{amsmath,amsthm,amsfonts,amssymb,amsbsy,pstricks,bbm,mathrsfs,multirow,bigstrut,graphicx,fixmath,url}
\usepackage[english]{babel} 
\usepackage{proof} 
\usepackage{stmaryrd}
\usepackage{xcolor}
\usepackage{todonotes}
\usepackage{lipsum}

\newcommand\blfootnote[1]{%
\begingroup
\renewcommand\thefootnote{}\footnote{#1}%
\addtocounter{footnote}{-1}%
\endgroup
}

\newcommand{\concretebags}{\mathcal{B}}
\newcommand{\automaton}{\mathcal{A}}
\newcommand{\powerset}{\mathcal{P}}

\newcommand{\isomorphism}{\mu}
\newcommand{\vertexiso}{\dot{\mu}}
\newcommand{\edgeiso}{\overline{\mu}}
\newcommand{\agraph}{G}
\newcommand{\edgeendpoints}{\mathit{endpts}}
\newcommand{\states}{Q}
\newcommand{\finalstates}{F}

\newcommand{\nodes}{\mathit{nodes}}
\newcommand{\arcs}{\mathit{arcs}}

\newcommand{\morphism}{\mu}

\newcommand{\subdecompositions}{\mathrm{Sub}}

\newcommand{\treewidth}{\mathrm{tw}}

\newtheoremstyle{named}{}{}{\itshape}{}{\bfseries}{.}{.5em}{Restatement of #1 \thmnote{#3}}
\theoremstyle{named}
\newtheorem*{retheorem}{Theorem}

\newcommand{\emptystring}{{\lambda}}
\newcommand{\positions}{\mathit{Pos}}
\newcommand{\myterms}{\mathit{Ter}}

\newcommand{\boldT}{\mathbf{T}}

\newcommand{\msotwo}{\mbox{CMSO}}
\newcommand{\cmso}{\mbox{CMSO}}

\newcommand{\treeAutomaton}{{\mathcal{A}}}
\newcommand{\lang}{{\mathcal{L}}} 
\newcommand{\N}{{\mathbb{N}}}

\newcommand{\projection}{\mathbold{\pi}}
\newcommand{\astate}{{\mathfrak{q}}}

\newcommand{\graphparameter}{\mathbf{p}}

\newcommand{\composedT}{%
  \mathrel{\vbox{\offinterlineskip\ialign{%
    \hfil##\hfil\cr
    $\scriptscriptstyle\circ$\cr
    \noalign{\kern0.1ex}
    $\boldT$\cr
}}}}

\newcommand{\composedTprime}{%
  \mathrel{\vbox{\offinterlineskip\ialign{%
    \hfil##\hfil\cr
    $\scriptscriptstyle\circ$\cr
    \noalign{\kern0.1ex}
    $\boldT'$\cr
}}}}

\newcommand{\graph}{{\mathcal{G}}}
\newcommand{\alphabet}{\Sigma}

\newcommand{\transitionsautomaton}{\Delta}
\newcommand{\algorithm}{\mathfrak{A}}

\newcommand{\maximumdegree}{\Delta}
\newcommand{\treewidthinput}{t}

\newcommand{\incidence}{\mathrm{Inc}}


\begin{document}

\title{On Supergraphs Satisfying $\cmso$ 
Properties}

\author[Mateus de Oliveira Oliveira]{Mateus de Oliveira Oliveira}
\address{Department of Informatics - University of Bergen, Bergen, Norway}
\email{mateus.oliveira@uib.no}

\begin{abstract}
Let $\cmso$ denote the counting monadic second-order logic of graphs.
We give a constructive proof that for some computable function $f$, 
there is an algorithm $\algorithm$ that takes as input a $\cmso$ sentence $\varphi$, 
a positive integer $t$, and a connected graph $\agraph$ of maximum degree at most $\Delta$,
and determines, in time 
$f(|\varphi|,\treewidthinput)\cdot 2^{O(\maximumdegree\cdot \treewidthinput)}\cdot |\agraph|^{O(\treewidthinput)}$, 
whether $\agraph$ has a supergraph $\agraph'$ of treewidth at most $t$ such that $\agraph'\models \varphi$.

The algorithmic metatheorem described above sheds new light on certain unresolved 
questions within the framework of graph completion algorithms. In particular, using this metatheorem, 
we provide an explicit algorithm that determines, in time 
$f(d)\cdot 2^{O(\maximumdegree \cdot d)}\cdot |\agraph|^{O(d)}$, whether a connected graph of maximum degree $\Delta$
has a planar supergraph of diameter at most $d$. Additionally, we show that for each fixed $k$, 
the problem of determining whether $\agraph$ has an $k$-outerplanar supergraph of diameter at most $d$ 
is strongly uniformly fixed parameter tractable with respect to the parameter $d$. 

This result can be generalized in two directions. First, the diameter parameter can be 
replaced by any contraction-closed effectively CMSO-definable parameter $\graphparameter$.
Examples of such parameters are vertex-cover number, dominating number, and many other contraction-bidimensional 
parameters. In the second direction, the planarity requirement can be relaxed to bounded genus, 
and more generally, to bounded local treewidth. 
\end{abstract}

\maketitle

\section{Introduction}
\label{section:Introduction}
\blfootnote{A preliminary version of this work appeared at
the 26th EACSL Annual Conference on Computer Science Logic, CSL 2017 \cite{deOliveiraOliveira2017Supergraphs}.}
A parameterized problem $\lang \subseteq \alphabet^*\times \N$ is said to be 
{\em fixed parameter tractable} (FPT) if there exists a function $f:\N\rightarrow \N$ such that 
for each $(x,k)\in \alphabet^*\times \N$, one can decide whether $(x,k)\in \lang$ in time 
$f(k)\cdot |x|^{O(1)}$, where $|x|$ is the size of $x$ \cite{DowneyFellows1999Book}. 
Using non-constructive methods derived from Robertson and Seymour's graph minor theory, 
one can show that certain problems can be solved in time $f(k)\cdot |x|^{O(1)}$ for some function 
$f:\N\rightarrow \N$. The caveat is that the function $f$ arising from these non-constructive methods
is often {\em not known to be computable}. Interestingly, for some problems it is not even clear how to obtain
algorithms running in time $f_1(k)\cdot |x|^{f_2(k)}$ for some {\em computable} functions $f_1$ and $f_2$. 
In this work we will use techniques from automata theory and structural graph theory
to provide constructive FPT and XP algorithms for problems for which only non-constructive parameterized algorithms were known.

The counting monadic second-order logic of graphs $(\msotwo)$ extends first order logic by allowing 
quantifications over sets of vertices and sets of edges, and by introducing the notion of modular counting predicates. 
This logic is expressive enough to define several interesting graph properties, such as 
Hamiltonicity, $3$-colorability, connectivity, planarity, fixed genus, minor embeddability, etc.
Additionally, when restricted to graphs of constant treewidth, CMSO logic is able to define precisely 
those properties that are recognizable by finite state tree-automata operating on encodings of tree-decompositions, 
or equivalently, those properties that can be described by equivalence relations with  finite index
\cite{Courcelle1990MSO,Abrahamson1993finite,PilipczukBojanczyk2016,PilipczukBojanczyk2017}.

The expressiveness of CMSO logic has had a great impact in algorithmic theory due to Courcelle's 
model-checking theorem \cite{Courcelle1990MSO}. This theorem states that for some computable function 
$f:\N^{2}\rightarrow \N$, one can determine in time\footnote{$|\agraph|$ denotes the number of vertices in $\agraph$, and 
$|\varphi|$, the number of symbols in $\varphi$.} $f(|\varphi|,t)\cdot |\agraph|$ whether a given 
graph $\agraph$ of treewidth at most $t$ satisfies a given $\cmso$ sentence $\varphi$. As a consequence 
of Courcelle's theorem, many combinatorial problems, such as Hamiltonicity or $3$-colorability, which are NP-hard on general graphs, 
can be solved in linear time on graphs of constant treewidth. 
In this work we will consider a class problems on graphs of constant treewidth which cannot be 
directly addressed via Courcelle's theorem, either because it is not clear how to formulate 
the set of positive instances of such a problem as a $\cmso$-definable set, or because although the 
set of positive instances is $\cmso$-definable, it is not clear how to explicitly construct a 
$\cmso$ sentence $\varphi$ defining such set. For instance, 
sets of graphs that are closed under minors very often fall in the second category due to Robertson and Seymour's graph 
minor theorem.

\subsection{Main Result}

Let $\varphi$ be a $\cmso$ sentence, and $t$ be a positive integer. We say that 
a graph $G'$ is a $(\varphi,t)$-supergraph of a graph $G$ if the following conditions are 
satisfied: $G'$ satisfies $\varphi$, $G'$ has treewidth at most $t$, and $G'$ is a supergraph 
of $G$ (possibly containing more vertices than $G$). 

In our main result, Theorem \ref{theorem:MainTheoremGraphCompletion}, 
we devise an algorithm $\algorithm$ that takes as input a $\cmso$ sentence $\varphi$, 
a positive integer $t$, and a connected graph $\agraph$ of maximum degree $\maximumdegree$, and determines in time 
$f(|\varphi|,t)\cdot 2^{O(\maximumdegree\cdot t)}\cdot |\agraph|^{O(t)}$ whether $\agraph$ has a $(\varphi,t)$-supergraph.
We note that our algorithm determines the existence of such a $(\varphi,t)$-supergraph $\agraph'$ without the 
need of necessarily constructing $\agraph'$. Therefore, no bound on the size of a candidate supergraph $\agraph'$ is 
imposed. Note that a priori even the fact that the problem is decidable is not clear.  

In the next three sub-sections we show how Theorem \ref{theorem:MainTheoremGraphCompletion} can be used to 
provide partial solutions to certain long-standing open problems in parameterized complexity theory.

\subsection{Planar Diameter Improvement}
\label{subsection:PDI}

In the {\sc planar diameter improvement} problem (PDI), we are given a graph $\agraph$, and 
a positive integer $d$, and the goal is to determine whether $G$ has a planar supergraph $G'$ of 
diameter at most $d$. Note that the set of YES instances for the PDI problem is closed under minors. 
In other words, if $\agraph$ has a planar supergraph of diameter at most $d$, then any minor 
$H$ of $\agraph$ also has such a supergraph. Therefore, using non-constructive arguments from 
Robertson and Seymour's graph minor theory \cite{RobertsonSeymour1995GraphMinorsXIII,RobertsonSeymour2004GraphMinorsXX}
in conjunction with the fact planar graphs of constant diameter have constant treewidth, one can show
that for each fixed $d$, there exists an algorithm $\algorithm_d$ which determines in linear 
time whether a given $\agraph$ has diameter at most $d$. The problem is that the non-constructive techniques mentioned 
above provide us with no clue about what the algorithm $\algorithm_d$ actually is. This problem can be partially
remedied using a technique called effectivization by self-reduction introduced by Fellows and 
Langston \cite{FellowsLangston1989,DowneyFellows1999Book}.
Using this technique one can show that for some function $f:\N\rightarrow\N$, there exists a single algorithm $\algorithm$
which takes a graph $\agraph$ and a positive integer $d$ as input, and determines in time $f(d)\cdot |\agraph|^{O(1)}$ 
whether $\agraph$ has a planar supergraph of diameter at most $d$. 
The caveat is that the function $f:\N\rightarrow \N$ bounding the influence of the parameter $d$ in the 
running time of the algorithm mentioned above is not known to be computable.

Obtaining a fixed parameter tractability result for the PDI problem with a {\em computable} function $f$ 
is a notorious and long-standing open problem in parameterized complexity theory 
\cite{DowneyFellows1999Book,FeasibleMathematics1995,CyganFominKowalik2015ParameterizedAlgorithmsBook}.
Indeed, when it comes to explicit algorithms, the status of the PDI problem is much more elusive. As 
remarked in \cite{CohenGoncalvesKimPaulSauThilikosWeller15}, even the problem of determining whether PDI 
can be solved in time $f_1(d)\cdot |G|^{f_2(d)}$ for {\em computable} functions $f_1,f_2:\N\rightarrow \N$ is open.
 
Using Theorem \ref{theorem:MainTheoremGraphCompletion} we provide an explicit algorithm that solves the 
PDI problem for connected graphs in time ${f(d)\cdot 2^{O(\maximumdegree\cdot d)}\cdot |\agraph|^{O(d)}}$ 
where $f:\N\rightarrow \N$ is a {\em computable} function, and $\maximumdegree$ is the maximum degree of $\agraph$.
This result settles an open problem stated in \cite{CohenGoncalvesKimPaulSauThilikosWeller15} 
in the case in which the input graph is connected and has bounded (even logarithmic) degree. We note that our 
algorithm imposes no restriction on the degree of a prospective supergraph $\agraph'$. 

\subsection{$k$-Outerplanar Diameter Improvement}
\label{subsection:OPDI}

A graph is $1$-outerplanar if it can be embedded in the plane in such a way that all vertices lie in the 
outer-face of the embedding. A graph is $k$-outerplanar if it can be embedded in the plane in such a way 
that that deleting all vertices in the outer-face of the embedding yields a $(k-1)$-outerplanar graph. 
The $k$-outerplanar diameter improvement problem ($k$-OPDI)
is the straightforward variant of PDI in which the completion is required to be $k$-outerplanar instead of planar.
In \cite{CohenGoncalvesKimPaulSauThilikosWeller15} Cohen at al. devised an explicit polynomial time algorithm for the 
$1$-OPDI problem. The complexity of the $k$-outerplanar diameter improvement problem was left open for 
$k\geq 2$. Using Theorem \ref{theorem:MainTheoremGraphCompletion} we show that the $k$-OPDI problem can 
be solved in time $f(k,d)\cdot 2^{O(\Delta\cdot k)}\cdot |G|^{O(k)}$ where $f:\N\times\N\rightarrow \N$ 
is a computable function. In other words, for each fixed $k$, the $k$-outerplanar diameter improvement 
problem is strongly uniformly fixed parameter tractable with respect to the diameter parameter $d$ for 
bounded degree connected input graphs.

\subsection{Contraction-Closed Parameters}

A graph parameter is a function $\graphparameter$ that associates a non-negative integer with each 
graph. We say that such a parameter is contraction-closed if $\graphparameter(G)\leq \graphparameter(G')$
whenever $G$ is a contraction of $G'$. For instance, the diameter of a graph is clearly a contraction-closed 
parameter. We say that a graph parameter $\graphparameter$ is effectively CMSO-definable if there exists a computable 
function $\alpha$, and an algorithm that takes a positive integer $k$ as input and constructs a CMSO 
formula $\varphi_k$ that is true on a graph $\agraph$ if and only if $\graphparameter(\agraph)\leq k$. 

The results described in the previous subsections can be generalized in two directions. First, the diameter 
parameter can be replaced by any effectively CMSO-definable contraction closed parameter that is unbounded 
on Gamma graphs. These graphs were defined in \cite{FominGolovachThilikos2011Contraction} with the goal to provide 
a simplified exposition of the theory of contraction-bidimensionality. In particular, many well studied parameters that 
arise often in bidimensionality theory satisfy the conditions listed above. Examples of such parameters are the 
sizes of a minimum vertex cover, feedback vertex set,  maximal matching, dominating set, edge dominating set, 
connected dominating set etc. On the other direction, the planarity requirement can be relaxed to CMSO definable 
graph properties that exclude some apex graph as a minor. These properties are also known in the literature 
as bounded local-treewidth properties. For instance, embeddability on surfaces of genus $g$, for fixed $g$,
is one of such properties.

\subsection{Related Work}

As mentioned above, given a CMSO sentence $\psi$ and a positive integer $t$, one can use Courcelle's model 
checking theorem to determine in time $f(|\psi|,t)\cdot |G|^{O(1)}$ whether a given graph $G$ of treewidth
at most $t$ satisfies $\psi$. Therefore, given a CMSO sentences $\varphi$ and a positive integer $t$,
we may consider the following algorithmic approach to decide whether a given graph $G$ has a $(\varphi,t)$-supergraph: 
first, we construct a formula $\psi_{\varphi,t}$ which is true on a graph $G$ if there is a model $G'$ of $\varphi$ of 
treewidth at most $t$ such that $G$ is a subgraph of $G'$. In other words, $\psi_{\varphi,t}$ defines the subgraph
closure of the set of models of $\varphi$ of treewidth at most $t$. Then, to determine whether $G$ has a $(\varphi,t)$-supergraph,
it is enough to determine whether $G$ satisfies $\psi_{\varphi,t}$ using Courcelle's model checking theorem.

Unfortunately, this approach cannot work in general. The problem is that there exist CMSO definable families of 
graphs whose subgraph closure is not CMSO definable. For instance, let $\mathcal{L} = \{L_n\}_{n\in \N}$ be 
the family of ladder graphs, where
$L_n$ is the ladder with $n$ steps\footnote{The vertices of $L_n$ are $a_1,...,a_n$ and $b_1,...,b_n$, and the edges 
are $\{a_i,b_i\}$ for $i\in [n]$, $\{a_i,a_{i+1}\}$ for $i\in [n-1]$, and $\{b_i,b_{i+1}\}$ for $i\in [n-1]$.}.
It is easy to see that $\mathcal{L}$ is CMSO definable and every graph in $\mathcal{L}$
has treewidth at most $2$. Nevertheless, the subgraph closure of $\mathcal{L}$ does not have finite index. 
Therefore, this subgraph closure is {\em not} CMSO definable, since CMSO definable classes of graphs of constant
treewidth have finite index. 

Interestingly, when the property defined by $\varphi$ is contraction closed, then the sentence $\psi_{\varphi,t}$ defines 
a minor-closed property $\mathcal{P}$ whose treewidth is bounded by $t$. Additionally, it follows from Robertson and 
Seymour graph minor theorem that each minor-closed property can be characterized by a finite set $\mathcal{M}$ 
of forbidden minors. Therefore, if we were able to enumerate the minors in $\mathcal{M}$ constructively, we would 
immediately obtain a constructive polynomial time algorithm for determining whether a given graph $G$ has a $(\varphi,t)$-supergraph.  
It is worth noting that Adler, Kreutzer and Grohe have shown that if a minor-free graph property $\mathcal{P}$ is 
MSO definable and has constant treewidth, then one can effectively enumerate the set of forbidden minors 
for $\mathcal{P}$ \cite{AdlerGroheKreutzer2008}. In particular, by giving the sentence  $\psi_{\varphi,t}$ as input 
to the  algorithm in \cite{AdlerGroheKreutzer2008} we would get a list of forbidden minors characterizing the 
set of graphs that have a $(\varphi,t)$-supergraph. Nevertheless, the problem with this approach is that it 
is not clear how the sentence $\psi_{\varphi,t}$ can be constructed from $\varphi$ and $t$.

In the {\em embedded planar diameter improvement problem} (EPDI), the input consists of a planar graph $G$ embedded in the
plane, and a positive integer $d$. The goal is to determine whether one can add edges to the faces of this embedding 
in such a way that the resulting graph has diameter at most $d$. The difference between this problem and 
the PDI problem mentioned above is that in the EPDI problem, an embedding is given at the input, and edges must be added 
in such a way that the embedding is preserved, while in the PDI problem, no embedding is given at the input.
Recently, it was shown in \cite{LokshtanovOliveiraSaurabh2018} that EPDI for $n$-vertex graphs can be solved 
in time $2^{d^{O(d)}}n^{O(d)}$, while the analogous embedded problem for $k$-outerplanar graphs can be solved in time $2^{d^{O(d)}}n^{O(k)}$.

It is worth noting that the algorithms in \cite{LokshtanovOliveiraSaurabh2018} heavily exploit the embedding of the 
input graph by viewing separators as {\em nooses} - simple closed curves in the plane that touch the graph only in the
vertices (see e.g. \cite{Bouchitte2003chordal}). Additionally, it is currently unknown both whether PDI can be reduced to EPDI 
in XP time and whether EPDI can be reduced in XP time to PDI. Therefore it is not clear if the algorithm for EPDI can be used to provide a strongly uniform 
XP algorithm for PDI on general graphs. It is also worth noting that no hardness results for either PDI or EPDI are known. Indeed, 
determining whether either of these problems is NP-hard is also a long-standing open problem.

\subsection{Proof Sketch And Organization of the Paper}
\label{proof:Techniques}

In Section \ref{section:Preliminaries} we state some preliminary definitions. In 
Section \ref{section:ConcreteTreeDecompositions} we define the notions of concrete bags, and 
concrete tree decompositions. Intuitively, a concrete tree-decomposition is an algebraic 
structure that represents a graph together with one of its tree decompositions.
Using such structures we are able to define infinite families of graphs via tree-automata 
that accept infinite sets of tree decompositions. In particular, Courcelle's theorem 
can be transposed to this setting. More precisely, there is a computable function $f$ such that 
for each $\cmso$ sentence $\varphi$ and each $t\in \N$, one can construct in time 
$f(|\varphi|,t)$ a tree automaton $\treeAutomaton(\varphi,t)$ which accept precisely those concrete tree decompositions 
of width at most $t$ that give rise to graphs satisfying $\varphi$ (Theorem \ref{theorem:AutomatonMSO}).

In Section \ref{section:Subdecompositions} we define the notion of sub-decomposition 
of a concrete tree decomposition. Intuitively, if a concrete tree decomposition $\boldT$ represents 
a graph $G$, then a sub-decomposition of $\boldT$ represents a sub-graph of $G$. 
We show that given a tree-automaton $\treeAutomaton$ accepting a set $\lang(\treeAutomaton)$ of 
concrete tree decompositions, one can construct a tree automaton $\subdecompositions(\treeAutomaton)$ 
which accepts precisely those sub-decompositions of concrete tree decompositions in $\lang(\treeAutomaton)$
(Theorem \ref{theorem:SubdecompositionLemma}).

In Section \ref{section:AllDecompositions}, we introduce the main technical tool of this work. 
More specifically, we show that for each connected graph $\agraph$ of maximum degree $\Delta$, 
one an construct in time $2^{O(\Delta \cdot t)}\cdot |G|^{O(t)}$ a tree-automaton $\treeAutomaton(G,t)$
whose language $\lang(\treeAutomaton(G,t))$ consists precisely of those concrete tree decompositions of 
width at most $t$ that give rise to $G$ (Theorem \ref{theorem:AllDecompositions}).   

In Section \ref{section:SupergraphsAndCompletions} we argue that the problem of determining whether $\agraph$
has a supergraph of treewidth at most $t$ satisfying $\varphi$ is equivalent to determining whether 
the intersection of $\lang(\treeAutomaton(\agraph,t+1))$ with $\lang(\subdecompositions(\treeAutomaton(\varphi,t+1)))$ 
is non-empty. By combining Theorems \ref{theorem:AutomatonMSO}, \ref{theorem:SubdecompositionLemma} 
and \ref{theorem:AllDecompositions}, 
we infer that this problem can be solved in time ${f(|\varphi|,t)\cdot 2^{O(\Delta \cdot t)}\cdot |G|^{O(t)}}$ 
(Theorem \ref{theorem:MainTheoremGraphCompletion}). Finally, in Section \ref{section:ContractionClosedParameters}, 
we apply Theorem \ref{theorem:MainTheoremGraphCompletion} to obtain explicit algorithms for 
several supergraph problems involving contraction-closed parameters.

\section{Preliminaries}
\label{section:Preliminaries}

For each $n\in \N$, we let $[n]=\{1,...,n\}$. We let $[0] = \emptyset$. For each finite set $U$, we let $\powerset(U)$
denote the set of subsets of $U$.  For each $r\in \N$ and each finite set $U$, we let 
$\powerset^{\leq}(U,r) = \{U'\subseteq U \mid |U'|\leq r\}$ be the set of subsets of $U$ of size at most 
$r$, and $\powerset^{=}(U,r) = \{U'\subseteq U \mid  |U'| = r\}$ be the set of subsets of $X$ of size
precisely $r$. If $A,A_1,...,A_k$ are sets, then we write 
$A = A_1\;\dot\cup\;A_2\;\dot\cup ... \dot\cup\;A_k$ to indicate that $A_i\cap A_j = \emptyset$ for $i\neq j$, 
and that $A$ is the disjoint union of $A_1,..,A_k$.

\subparagraph*{Graphs:} A {\em graph} is a triple $G = (V_G,E_G,\incidence_G)$ where $V_G$ is 
a set of vertices, $E_G$ is a set of edges, and $\incidence_G\subseteq E_G\times V_G$ is an 
incidence relation. For each $e\in E_G$ we let $\edgeendpoints(e) = \{v \mid \incidence_G(e,v)\}$ 
be the set of endpoints of $e$, and we assume that $|\edgeendpoints(e)|$ is either $0$ or $2$. 
We note that our graphs are allowed to have multiple edges, but no loops. 
We say that a graph $H$ is a subgraph of $G$ if $V_H\subseteq V_G$, $E_H\subseteq E_G$ and 
$\incidence_{H} = \incidence_G\cap E_H\times V_H$. Alternatively, we say that
$G$ is a supergraph of $H$. The degree of a vertex $v\in V_G$ is the number $d(v)$ of edges 
incident with $v$. We let $\maximumdegree(G)$ denote the maximum degree of a vertex of $G$.  

A {\em path} in a graph $G$ is a sequence $v_1e_1v_2...e_{n-1}v_{n}$ 
where $v_i\in V_G$ for $i\in [n]$, $e_i\in E_G$ for $i\in [n-1]$, $v_i\neq v_j$ for $i\neq j$, and 
$\{v_i,v_{i+1}\} = \edgeendpoints(e_i)$ for each $i\in [n-1]$. We say that $G$ is {\em connected} if for 
every two vertices $v,v'\in V_G$ there is a path whose first vertex is $v$ and whose last vertex is 
$v'$.

Let $G$ and $H$ be graphs. An isomorphism from $G$ to $H$ is a pair of bijections 
$\isomorphism = (\vertexiso:V_G\rightarrow V_H,\edgeiso:E_G\rightarrow E_H)$ such that for every $e\in E_G$
if $\edgeendpoints(e) = \{v,v'\}$ then $\edgeendpoints(\edgeiso(e)) = \{\vertexiso(v),\vertexiso(v')\}$. 
We say that $G$ and $H$ are isomorphic if there is an isomorphism from $G$ to $H$.

\subparagraph*{Treewidth:} A tree is an acyclic graph $T$ containing a unique connected component.
To avoid confusion we may call the vertices of a tree ``nodes'' and call their edges ``arcs''. 
We let $\nodes(T)$ denote the set of nodes of $T$ and $\arcs(T)$ denote its set of arcs. A {\em tree decomposition} 
of a graph $G$ is a pair $(T,\beta)$ where $T$ is a tree and $\beta:\nodes(T)\rightarrow \powerset(V_G)$
is a function that labels nodes of $T$ with subsets of vertices of $G$ in such a way that 
the following conditions are satisfied. 

\begin{enumerate}
\item $\bigcup_{u\in \nodes(T)} \beta(u) = V_G$ 
\item For every $e \in E_G$, there exists a node $u\in \nodes(T)$ such that $\edgeendpoints(e) \subseteq \beta(u)$ 
\item For every $v\in V_G$, the set $T_v = \{u\in \nodes(T)\mid v\in \beta(u)\}$, i.e., the set of nodes of $T$ 
	whose corresponding bags contain $v$, induces a connected subtree of $T$. 
\end{enumerate}

The {\em width} of a tree decomposition $(T,\beta)$ is defined as $\mathit{max}_{u\in \nodes(T)} |\beta(u)|-1$, 
that is, the maximum bag size minus one. The treewidth of a graph $G$, denoted by $\treewidth(G)$, is the minimum 
width of a tree decomposition of $G$.

\subparagraph*{$\cmso$ Logic:} 
The counting monadic second-order logic of graphs, here denoted by 
$\cmso$, extends first order logic by allowing quantifications 
over sets of vertices and edges, and by introducing the notion of modular counting predicates. 
More precisely, the syntax of $\cmso$ logic includes the logical connectives 
$\vee,\wedge,\neg,\Leftrightarrow,\Rightarrow$, variables for vertices, edges, sets of vertices 
and sets of edges, the quantifiers $\exists,\forall$ that can be applied to these variables, and the following atomic predicates:
\begin{enumerate}
	\item $x\in X$ where $x$ is a vertex variable and $X$ a vertex-set variable;
	\item $y\in Y$ where $y$ is an edge variable and $Y$ an edge-set variable;
	\item $\incidence(x,y)$ where $x$ is a vertex variable, $y$ is an edge variable, and the interpretation is 
			that the edge $x$ is incident with the edge $y$. 
	\item $\mathit{card}_{a,r}(Z)$ where $0\leq a < r$, $r\geq 2$, $Z$ is a vertex-set or edge-set variable,
		 and the interpretation is that $|Z| = a\;\; (\!\!\!\mod r)$;
	\item equality of variables representing vertices, edges, sets of vertices and sets of edges. 
\end{enumerate}

A $\cmso$ {\em sentence} is a $\cmso$ formula without free 
variables. If $\varphi$ is a $\cmso$ sentence, then we write 
$G\models \varphi$ to indicate that $G$ satisfies $\varphi$. 

\subparagraph*{Terms:} Let $\alphabet$ be a finite set. The set $\myterms(\alphabet)$ of terms 
over $\alphabet$ is inductively defined as follows. 
\begin{enumerate}
	\item If $a\in \alphabet$, then $a\in \myterms(\alphabet)$.
	\item If $a\in \alphabet$, and $t\in \myterms(\alphabet)$, then $a(t)\in \myterms(\alphabet)$. 
	\item If $a\in \alphabet$, and $t_1,t_2\in \myterms(\alphabet)$, then $a(t_1,t_2)\in \myterms(\alphabet)$.  
\end{enumerate}

Note that the alphabet $\alphabet$ is unranked and the symbols in $\alphabet$ may be regarded as 
function symbols or arity $0$, $1$ or $2$. The set of positions of a term $t=a(t_1,...,t_r)\in \myterms(\alphabet)$
is defined as follows.
$$\positions(t)=\{\emptystring\} \cup \bigcup_{ j \in \{1,...,r\}} \{ j p\;|\; p\in \positions(t_j)\}.$$

Note that $\positions(t)$ is a set of strings over the alphabet $\{1,2\}$ and that $\emptystring$ is the empty string. 
If $t=a$ for some $a\in \alphabet$, then $\positions(t) = \{\lambda\}$.
If $p,pj\in \positions(t)$ where $j\in \{1,2\}$, then we say that $pj$ is a {\em child} of $p$.
Alternatively, we say that $p$ is the {\em parent} of $pj$. 
We say that $p$ is a {\em leaf} if it has no children.  
We let $\tau(t)$ be the tree that has 
$\positions(t)$ as nodes and $\{ \{p,pj\} \mid j\in \{1,2\},\; p,pj\in \positions(t)\}$ as arcs.
We say that a subset $P\subseteq \positions(t)$ is {\em connected} if the sub-tree of $\tau(t)$ 
induced by $P$ is connected. If $P$ is connected, then we say that a position $p\in P$ is 
the {\rm root} of $P$ if the parent of $p$ does not belong to $P$. 

If $t=a(t_1,...,t_r)$ is a term in $\myterms(\alphabet)$ for $r\in \{0,1,2\}$, and $p\in \positions(t)$, 
then the symbol $t[p]$ at position $p$ is inductively defined as follows. 
If $p=\emptystring$, then $t[p]=a$. On the other hand, if $p=jp'$ where $j\in \{1,2\}$,
then $t[p] = t_j[p']$.

\subparagraph*{Tree Automata:}
Let $\alphabet$ be a finite set of symbols. 
A {\em bottom-up tree-automaton} over
$\alphabet$ is a tuple $\treeAutomaton=(Q,\alphabet,\finalstates,\transitionsautomaton)$
where $Q$ is a set of states, $\finalstates\subseteq Q$ a set of final states, and 
$\transitionsautomaton$ is a set of transitions of the form $(\astate_1,...,\astate_r,a,\astate)$ 
with $a\in \alphabet$, $0\leq r\leq 2$,  and $\astate_1,...,\astate_r,\astate \in Q$. 
The size of $\treeAutomaton$, which is defined as $|\treeAutomaton|=|Q|+|\transitionsautomaton|$, 
measures the number of states in $Q$ plus the number of transitions in $\transitionsautomaton$.
The set $\lang(\treeAutomaton,\astate,i)$ of all terms reaching a state $\astate\in Q$ in depth 
at most $i$ is inductively defined as follows.

\begin{equation*}
\label{equation:InductiveDefinition}
\begin{array}{l}
\lang(\treeAutomaton,\astate,1) = \{a\; |\; (a,\astate)\in \transitionsautomaton\} \\
\\
\lang(\treeAutomaton,\astate,i) = \lang(\treeAutomaton,\astate,i-1) \; \cup \; \\
\hspace{1.9cm} \{a(t_1,...,t_{r})\;|\; r\in \{1,2\}, \mbox{ and } 
\exists (\astate_1,...,\astate_{r},a,\astate) \in \transitionsautomaton,\; 
t_j\in \lang(\treeAutomaton,\astate_j,i-1)\}
\end{array}
\end{equation*}

We denote by $\lang(\treeAutomaton,\astate)$ the set of all terms reaching state $\astate$ in finite depth, 
and by $\lang(\treeAutomaton)$ the set of all terms reaching some final state in $\finalstates$. 

\begin{equation}
\lang(\treeAutomaton,\astate) =\bigcup_{i\in \N} \lang(\treeAutomaton,\astate,i)
\hspace{1.5cm}
\lang(\treeAutomaton) =\bigcup_{{\astate\in \finalstates}} \lang(\treeAutomaton,\astate)
\end{equation}

We say that the set $\lang(\treeAutomaton)$ is the language {\em accepted} by $\treeAutomaton$.

Let $\projection:\alphabet\rightarrow \alphabet'$ be a map between finite sets of symbols $\alphabet$ and $\alphabet'$. 
Such mapping can be homomorphically extended to a mapping $\projection:\myterms(\alphabet)\rightarrow \myterms(\alphabet')$ 
between terms by setting $\projection(t)[p]=\projection(t[p])$ for each position $p\in \positions(t)$. 
Additionally, $\projection$ can be further extended to a set of terms $\lang\subseteq \myterms(\alphabet)$ by setting 
${\projection(\lang) = \{\projection(t)\;|\;t\in \myterms(\alphabet)\}}$. Below we state some well 
known closure and decidability properties for tree automata. 

\begin{lem}[Properties of Tree Automata \cite{Tata2007}]
\label{lemma:PropertiesOfTreeAutomata}
\label{lemma:EmptynessIntersection}
\label{lemma:Projection}
Let $\alphabet$ and $\alphabet'$ be finite sets of symbols. 
Let $\treeAutomaton_1$ and $\treeAutomaton_2$ be tree automata over $\alphabet$, 
and $\projection:\alphabet\rightarrow\alphabet'$ be a mapping. 
\begin{enumerate}
	\item One can construct in time $O(|\treeAutomaton_1|\cdot |\treeAutomaton_2|)$ a 
tree automaton $\treeAutomaton_1\cap \treeAutomaton_2$ such that 
$\lang(\treeAutomaton_1\cap \treeAutomaton_2) = \lang(\treeAutomaton_1)\cap \lang(\treeAutomaton_2)$. 
	\item One can determine whether $\lang(\treeAutomaton_1) \neq \emptyset$ in time $O(|\treeAutomaton_1|)$.  
	\item One can construct in time $O(|\treeAutomaton_1|)$ a tree automaton $\projection(\treeAutomaton_1)$
		such that $\lang(\projection(\treeAutomaton_1)) = \projection(\lang(\treeAutomaton_1))$. 
\end{enumerate}
\end{lem}

\section{Concrete Tree Decompositions}
\label{section:ConcreteTreeDecompositions}

\newcommand{\bagset}{B}
\newcommand{\bagedge}{b}

A {\em  $t$-concrete bag } is a pair $(\bagset,\bagedge)$ where $\bagset\subseteq [t]$, and 
$\bagedge\subseteq \bagset$ with $\bagedge = \emptyset$ or $|\bagedge| =  2$. We note that $B$
is allowed to be empty. We let $\concretebags(t)$ be the set of all $t$-concrete bags. 
Note that $|\concretebags(t)| \leq t^{2}\cdot 2^{t}$. We regard the set $\concretebags(t)$ as a 
finite alphabet which will be used to construct terms representing tree decompositions of graphs. 

A {$t$-concrete tree decomposition} is a term $\boldT\in \myterms(\concretebags(t))$. We 
let $\boldT[p]= (\boldT[p].\bagset,\boldT[p].\bagedge)$ be the $t$-concrete bag at position $p$ 
of $\boldT$. For each $s\in [t]$, we say that a non-empty subset $P\subseteq \positions(\boldT)$ is 
{\em $s$-maximal} if the following conditions are satisfied.

\begin{enumerate}
	\item $P$ is connected in $\positions(\boldT)$. 
	\item $s\in \boldT[p].\bagset$ for every $p\in P$.
	\item If $P'$ is a connected subset of $\positions(\boldT)$ and $s\in \boldT[p].\bagset$
		for every $p\in P'$, then either $P\cap P'=\emptyset$ or $P'\subseteq P$. 
\end{enumerate}

Note that if $P$ and $P'$ are $s$-maximal then either $P=P'$, or $P\cap P' =\emptyset$.  
Additionally, for each $p\in \positions(\boldT)$, and each $s\in \boldT[p].\bagset$, 
there exists a unique subset $P\subseteq \positions(\boldT)$ such that $P$ is $s$-maximal 
and $p\in P$. We denote this unique set by $P(p,s)$. Intuitively, each such set $P(p,s)$ 
corresponds to a vertex in the graph represented by $\boldT$. Two disjoint $s$-maximal sets $P(p,s)$ and $P(p',s)$ correspond to two distinct vertices in the graph.

\begin{defi}
\label{definition:GraphConcreteDecomposition}
Let $\boldT\in \myterms(\concretebags(t))$. 
The graph $\graph(\boldT)$ associated with $\boldT$ is defined as follows. 
\begin{enumerate}
	\item\label{graphOne} 
$V_{\graph(\boldT)}  = \{v_{s,P}\;|\;s\in [t],\; P\subseteq \positions(\boldT),\; P\mbox{ is $s$-maximal}\}$. 
	\item\label{graphTwo} $E_{\graph(\boldT)} = \{e_p \;|\; p\in \positions(\boldT),\; \bagedge\neq \emptyset\}$.
	\item\label{graphThree} 
$\incidence_{\graph(\boldT)}  = \{(e_p,v_{s,P(p,s)})\;|\; e_p\in E_{\graph(\boldT)},\; s\in \boldT[p].\bagedge\}$.  
\end{enumerate}
\end{defi}

Intuitively, a $t$-concrete tree decomposition may be regarded as a way of representing 
a graph together with one of its tree decompositions. This idea is widespread in texts 
dealing with recognizable properties of graphs 
\cite{PilipczukBojanczyk2016,AdlerGroheKreutzer2008,CourcelleEngelfriet2012,Elberfeld2016,Flum2002query}.
Within this framework it is customary to define a bag of width $t$ as a graph with at most 
$t$ vertices together with a function that labels the vertices of these graphs with numbers 
from $\{1,...,t\}$. Our notion of $t$-concrete bag, on the other hand, may be regarded as 
a representation of a graph with at most $t$ vertices injectively labeled with numbers from 
$\{1,...t\}$ and at most {\em one} edge. Within this point of view, the representation used
here is a syntactic restriction of the former. On the other hand, any decomposition which uses 
bags with arbitrary graphs of size $t$ can be converted into a $t$-concrete decomposition, by
expanding each bag into a sequence of $t^2$ concrete bags. The following observation is immediate, using 
the fact that if a graph has treewidth $t$, then it has a rooted tree decomposition in which each 
node has at most two children \cite{Elberfeld2016}.

\begin{obs}
\label{observation:TreeDecompositionsConcreteTreeDecompositions}
A graph $G$ has treewidth $t$ if and only if there exists some $(t+1)$-concrete tree decomposition 
$\boldT\in \myterms(\concretebags(t+1))$ such that $\graph(\boldT)$ is isomorphic to $G$. 
\end{obs}

The next theorem (Theorem \ref{theorem:AutomatonMSO}) may be regarded as a variant of Courcelle's theorem \cite{CourcelleEngelfriet2012}.
For completeness, we include a proof of Theorem \ref{theorem:AutomatonMSO} in Appendix \ref{ProofTheoremAutomatonMSO}.

\begin{thm}[Variant of Courcelle's Theorem]
\label{theorem:AutomatonMSO}
There exists a computable function $f:\N\times \N \rightarrow \N$ such that for each $\cmso$ sentence
$\varphi$, and each $t\in \N$, one can construct in time $f(|\varphi|,t)$ a tree-automaton 
$\automaton(\varphi,t)$ accepting the following tree language. 
\begin{equation}
\lang(\treeAutomaton(\varphi,t))  = \{\boldT\in \myterms(\concretebags(t))\;|\; \graph(\boldT) \models \varphi\}.
\end{equation}
\end{thm}

\section{Sub-Decompositions}
\label{section:Subdecompositions}

In this section we introduce the notion of sub-decompositions of a $t$-concrete decomposition. 
Intuitively, if a $t$-concrete tree decomposition $\boldT$ represents a graph $G$ then 
sub-decompositions of $\boldT$ represent subgraphs of $G$. The main result of this section states 
that given a tree automaton $\treeAutomaton$ over $\concretebags(t)$, one can efficiently construct 
a tree automaton $\subdecompositions(\treeAutomaton)$ over $\concretebags(t)$ which accepts precisely 
the sub-decompositions of $t$-concrete tree decompositions in $\lang(\treeAutomaton)$. 

We say that a $t$-concrete bag $(\bagset,\bagedge)$ is a sub-bag of a $t$-concrete bag $(\bagset',\bagedge')$ if 
$\bagset \subseteq \bagset'$ and $\bagedge\subseteq \bagedge'$. 

\begin{defi}
We say that a $t$-concrete tree decomposition $\boldT \in \myterms(\concretebags(t))$ 
is a sub-decomposition of a $t$-concrete tree decomposition $\boldT'\in \myterms(\concretebags(t))$ 
if the following conditions are satisfied. 

\begin{enumerate}[align=left]
\label{enumerate:SubDecomposition}
	\item[{\bf S1.}]\label{SubdecompositionOne} $\positions(\boldT) = \positions(\boldT')$. 
	\item[{\bf S2.}]\label{SubdecompositionTwo} For each $p\in \positions(\boldT)$, $\boldT[p]$ is a sub-bag of $\boldT'[p]$.
	\item[{\bf S3.}]\label{SubdecompositionThree} For each $p,pj\in \positions(\boldT)$, and for each $s\in [t]$, 
			if $s\in \boldT'[p].\bagset$ and $s\in \boldT'[pj].\bagset$, then $s\notin \boldT[p].\bagset$ 
			if and only if $s\notin \boldT[pj].\bagset$.
\end{enumerate}
\end{defi}

The following theorem states that sub-decompositions of $\boldT'$ are in one to one correspondence with 
subgraphs of $\graph(\boldT')$. 

\begin{thm}
\label{theorem:SubdecompositionLemma}
Let $G$ and $G'$ be graphs and let $\boldT'\in \myterms(\concretebags(t))$ be a $t$-concrete tree decomposition 
such that $\graph(\boldT')= G'$. 
Then $G$ is a subgraph of $G'$ if and only if there exists some $\boldT\in \myterms(\concretebags(t))$
such that $\boldT$ is a sub-decomposition of $\boldT'$ with $\graph(\boldT) = G$. 
\end{thm}
\begin{proof}
$ $ 
\begin{enumerate}
\item Let $G$ be a subgraph of $\graph(\boldT')$. We show that there exists a sub-decomposition $\boldT$ of 
$\boldT'$ such that $\graph(\boldT) = G$. Since $G$ is a subgraph of $\graph(\boldT)$, we have that 
$V_G \subseteq V_{\graph(\boldT')}$, $E_G\subseteq E_{\graph(\boldT')}$, and 
$\incidence_G = \incidence_{\graph(\boldT')}\cap E_G\times V_G$. We define $\boldT$ by setting 
$\boldT[p]$ as follows for each $p\in \positions(\boldT)=\positions(\boldT')$. 
\begin{enumerate}
	\item $\boldT[p].\bagset = \boldT'[p].\bagset 
	\backslash \{s\;|\; v_{s,P(p,s)} \in V_{\graph(\boldT')}\backslash V_G\}$.
	\item $\boldT[p].\bagedge = \emptyset$ if $e_p\in E_{\graph(\boldT')}\backslash E_{G}$ and 
	       $\boldT[p].\bagedge = \boldT'[p].\bagedge$ otherwise. 
\end{enumerate}
First, we note that $v_{s,P}\in V_{\graph(\boldT)}$ if and only if $v_{s,P}\in V_G$, $e_p\in E_{\graph(\boldT)}$ if 
and only if $e_p\in E_{G}$, and $(e_p,v_{i,P})\in \incidence_{\graph(\boldT)}$ if and only if 
$(e_p,v_{i,P})\in V_G$. Therefore, $G = \graph(\boldT)$. 
We will check that the $t$-concrete decomposition $\boldT$ defined above is indeed a sub-decomposition of $\boldT'$.
In other words, we will verify that conditions {\bf S1}, {\bf S2} and {\bf S3} above are satisfied. The fact 
that {\bf S1} is satisfied is immediate, since we define $\boldT[p]$ for each $p\in \positions(\boldT')$. Therefore, 
$\positions(\boldT) = \positions(\boldT')$. Condition {\bf S2} is also satisfied, since by $(a)$ and $(b)$ we have 
that $\boldT[p].\bagset\subseteq \boldT'[p].\bagset$ and that $\boldT[p].\bagedge$ is either $\emptyset$, or equal to 
$\boldT'[p].\bagedge$. Finally, condition {\bf S3} is also satisfied, since (a) guarantees that for each 
number $s\in [t]$, and each $s$-maximal set $P\subseteq \positions(\boldT')$, if $s$ is 
removed from $\boldT'[p].\bagset$ for some $p\in P$, then $s$ is indeed removed from $\boldT'[p].\bagset$ for every $p\in P$.

\item For the converse, let $\boldT$ be a sub-decomposition of $\boldT'$. We show that the graph 
	$\graph(\boldT)$ is a subgraph of $\graph(\boldT')$. First, we note that condition {\bf S3} guarantees 
	that for each $s\in [t]$ and each $P\subseteq \positions(\boldT)$, if $P$ is $s$-maximal in $\boldT$ 
	then $P$ is $s$-maximal in $\boldT'$. Therefore, $V_{\graph(\boldT)}\subseteq V_{\graph(\boldT')}$.
	Now, Condition {\bf S2} guarantees that $e_p\in E_{\graph(\boldT)}$ implies that $e_p\in E_{\graph(\boldT')}$. 
	Therefore, $E_{\graph(\boldT)}\subseteq E_{\graph(\boldT')}$. Finally, by definition 
	$(e_p,v_{s,P})\in \incidence_{\graph(\boldT)}$ if and only if $s\in \boldT[p].\bagedge$ for each $p\in P$. Since the fact 
	that $s\in \boldT[p].\bagedge$ implies that $s\in \boldT'[p].\bagedge$, we have that $(e_p,v_{s,P})\in \incidence_{\graph(\boldT)}$
	implies that $(e_p,v_{s,P})\in \incidence_{\graph(\boldT')}$. Therefore, 
	$\incidence_{\graph(\boldT)} \subseteq \incidence_{\graph(\boldT')}$. Additionally, since 
	$(e_p,v_{s,P(s,p)})\in \incidence_{\graph(\boldT)}$ for each $e_p\in E_{\graph(\boldT)}$ and each $s\in \boldT[p].\bagedge$,
	we have that $\incidence_{\graph(\boldT)}  = \incidence_{\graph(\boldT')} \cap E_{\graph(\boldT)}\times V_{\graph(\boldT)}$.
	This shows that $\graph(\boldT)$ is a subgraph of $\graph(\boldT')$. 
  \qedhere
\end{enumerate}
\end{proof}

The following theorem states that given a tree automaton $\treeAutomaton$ over $\concretebags(t)$, one can efficiently 
construct a tree automaton $\subdecompositions(\treeAutomaton)$ which accepts precisely the 
sub-decompositions of $t$-concrete tree decompositions in $\lang(\treeAutomaton)$. 

\begin{thm}[Sub-Decomposition Automaton]
\label{theorem:SubDecompositions}
Let $\treeAutomaton$ be a tree automaton over $\concretebags(t)$. Then one can construct in time 
$2^{O(t)}\cdot |\treeAutomaton|$ a tree automaton $\subdecompositions(\treeAutomaton)$ over 
$\concretebags(t)$ accepting the following language. 
\begin{equation*}
\lang(\subdecompositions(\treeAutomaton)) = \{\boldT\in \myterms(\concretebags(t))\;|\;
\exists \boldT'\in \lang(\treeAutomaton) \mbox{ s.t. $\boldT$ is a sub-decomposition of $\boldT'$} 
\}.
\end{equation*}
\end{thm}
\begin{proof}
Let $\treeAutomaton = (\states,\concretebags(t),\transitionsautomaton,\finalstates)$ be a tree automaton over $\concretebags(t)$.
As a first step we create an intermediate tree automaton 
$\treeAutomaton' = (\states',\concretebags(t),\transitionsautomaton',\finalstates')$ which accepts the 
same language as $\treeAutomaton$. The tree automaton $\treeAutomaton'$ is defined as follows.   
\begin{equation*}
\begin{array}{c}
\states' =  \{\astate_{\bagset}\;|\;\astate \in \states,\;\bagset\subseteq [t]\} \hspace{1cm} 
\finalstates' = \{\astate_{\bagset}\;|\; \astate\in \finalstates,\;\bagset \subseteq [t]\}
\\
\\
\transitionsautomaton' = \{(\astate_{\bagset_1}^1,...,\astate_{\bagset_r}^r,(\bagset,\bagedge),\astate_{\bagset})\;|\; 
(\astate^1,...,\astate^{r},(\bagset,\bagedge),\astate) \in \transitionsautomaton,\;B_i\subseteq [t] \mbox{ for $i\in [r]$}\}.
\end{array}
\end{equation*}
Note that for each $\astate\in \states$, each $\bagset\subseteq [t]$, and each 
$\boldT\in \myterms(\concretebags(t))$, $\boldT$ reaches state 
$\astate_{\bagset}$ in $\treeAutomaton'$ if and only if $\boldT$ reaches state $\astate$ in $\treeAutomaton$ and 
$\boldT[\emptystring].\bagset = \bagset$, where $\boldT[\emptystring]$ is the $t$-concrete bag at the root of $\boldT$. 
In particular, this implies that a term $\boldT$ belongs 
to $\lang(\treeAutomaton')$ if and only if $\boldT \in \lang(\treeAutomaton)$. 

Now, consider the tree automaton $\subdecompositions(\automaton) = (\states'',\concretebags(t),\transitionsautomaton'',\finalstates'')$ 
over $\concretebags(t)$ where 
\begin{equation*}
\begin{array}{l}
\states'' = \{\astate_{\bagset,\bagset'}\;|\;\astate\in \states, \bagset\subseteq \bagset'\subseteq [t]\} \hspace{1cm}
\finalstates'' = \{\astate_{\bagset,\bagset'}\;|\;\astate\in \finalstates, \bagset\subseteq \bagset'\subseteq [t]\}
\\
\\
\transitionsautomaton'' = \{(\astate^1_{\bagset_1,\bagset_1'},...,\astate^r_{\bagset_r,\bagset_r'},
(\bagset,\bagedge),\astate_{\bagset,\bagset'})\;|\;
\exists (\astate^1_{\bagset_1'},...,\astate^r_{\bagset_r'},(\bagset',\bagedge'),\astate_{\bagset'})\in \transitionsautomaton' 
\mbox{ such that}\;\\
\hspace{4.7cm}\bagset_i\subseteq \bagset_i',\;\bagset\subseteq \bagset',\\
\hspace{4.7cm}(\bagset,\bagedge) \mbox{ is a sub-bag of } (\bagset',\bagedge') \\
\hspace{4.7cm}\mbox{for each $j\in [r]$, if $s\in B'\wedge s\in B_j'$ then $s\notin B\Leftrightarrow s\notin B_j$} \}.
\end{array}
\end{equation*}

It follows by induction on the height of terms that 
a term $\boldT \in \myterms(\concretebags(t))$ reaches a state 
$\astate_{\bagset,\bagset'}$ in $\subdecompositions(\treeAutomaton)$ 
if and only if there exists some term $\boldT'\in \myterms(\concretebags(t))$ such that 
$\boldT'$ reaches state $\astate_{\bagset'}$ in $\treeAutomaton'$, $\boldT[\lambda].\bagset = \bagset$, 
$\boldT'[\lambda].\bagset = \bagset'$, and $\boldT$ is a sub-decomposition of $\boldT'$. In particular, 
$\boldT$ reaches a final state of $\subdecompositions(\treeAutomaton)$ if and only if $\boldT$
is a sub-decomposition of some $\boldT'$ which reaches a final state of $\treeAutomaton'$. 
\end{proof}

\section{Representing All Tree Decompositions of a Given Graph}
\label{section:AllDecompositions}

\newcommand{\vertexextension}{\nu}
\newcommand{\neighbourextension}{\eta}
\newcommand{\concretemorphism}{\mu}
\newcommand{\edgeextension}{y}
\newcommand{\rootmarking}{\rho}
\newcommand{\concreteprojection}{\mathbold{\pi}}

In this section we show that given a connected graph $G$ of maximum degree $\maximumdegree$, and 
a positive integer $t$, one can construct in time $2^{O(\Delta\cdot t)}\cdot |V_G|^{O(t)}$ a 
tree automaton $\treeAutomaton(G,t)$ over $\concretebags(t)$ that accepts the set of all 
$t$-concrete tree decompositions of $G$ (of all shapes and sizes).

Let $G$ be a graph. A $(G,t)$-concrete bag is a tuple 
$(\bagset,\bagedge,\vertexextension,\neighbourextension,\edgeextension,\rootmarking)$ where $(\bagset,\bagedge)$ 
is a $t$-concrete bag; ${\vertexextension:\bagset\rightarrow V_G}$ is an injective function that assigns a vertex 
of $G$ to each element of $\bagset$; ${\neighbourextension:\bagset\rightarrow \powerset^{\leq}(E_G,\maximumdegree(G))}$
is a function that assigns to each element $s\in \bagset$, a set of edges incident with $\vertexextension(s)$ of size 
at most $\maximumdegree(G)$; $\edgeextension$ is a subset of $E_G$ such that $|\edgeextension|\leq 1$ and 
$\edgeextension\subseteq \neighbourextension(s)$ whenever $s\in \bagedge$; and $\rootmarking$ is a subset of $\bagset$. 

We let $\concretebags(G,t)$ be the set of all $(G,t)$-concrete bags. 
Note that $\concretebags(G,t)$ has at most $2^{O(\maximumdegree(G)\cdot t)}\cdot |V_G|^{O(t)}$ elements. 
We let $\myterms(\concretebags(G,t))$ be the set of all terms over $\concretebags(G,t)$. If $\hat{\boldT}$
is a term in $\concretebags(G,t)$ then for each $p\in \positions(\boldT)$, the $(G,t)$-concrete 
bag of $\hat{\boldT}$ at position $p$ is denoted by the tuple
$$(\hat{\boldT}[p].\bagset,\hat{\boldT}[p].\bagedge,\hat{\boldT}[p].\vertexextension,
\hat{\boldT}[p].\neighbourextension,\hat{\boldT}[p].\edgeextension,\hat{\boldT}[p].\rootmarking).$$. 

\begin{defi}
\label{definition:GTDecomposition}
We say that a term $\hat{\boldT}\in \myterms(\concretebags(G,t))$ is a $(G,t)$-concrete tree decomposition 
if the following conditions are satisfied for each each $p\in \positions(\hat{\boldT})$ and each 
$s\in [t]$.
\begin{enumerate}[align=left,itemsep=3pt]
	\item[{\bf C1.}] If $pj\in \positions(\hat{\boldT})$ and $s\in \hat{\boldT}[p].\bagset \cap \hat{\boldT}[pj].\bagset$ 
		then ${\hat{\boldT}[p].\vertexextension(s) = \hat{\boldT}[pj].\vertexextension(s)}$. 
	\item[{\bf C2.}] If $\hat{\boldT}[p].\bagedge = \{s,s'\}$ then $\hat{\boldT}[p].\edgeextension = \{e\}$ for some 
		edge $e$ with $$\edgeendpoints(e) = \{\hat{\boldT}[p].\vertexextension(s),\hat{\boldT}[p].\vertexextension(s')\}.$$ 
	\item[{\bf C3.}] Let $r\in \{0,1,2\}$, and $p1,...,pr$ be the children\footnote{If $r=0$ then $p$ has no child.} of $p$, 
		then $$\hat{\boldT}[p].\neighbourextension(s) = \hat{\boldT}[p].\edgeextension \;\dot\cup\; 
			\hat{\boldT}[p1].\neighbourextension(s)\; \dot\cup\; ... \;\dot\cup\;\hat{\boldT}[pr].\neighbourextension(s).$$
	\item[\bf{C4.}] If $s\in \hat{\boldT}[p].\rootmarking$ then 
		$\hat{\boldT}[p].\neighbourextension(s) = \{e\;|\; (e,\hat{\boldT}[p].\vertexextension(s))\in \incidence_G\}$. 
	\item[\bf{C5.}] If $p=\emptystring$ then $\hat{\boldT}[p].\rootmarking = \hat{\boldT}[p].\bagset$. 
		If $pj\in \positions(\hat{\boldT})$ then $s\in \hat{\boldT}[pj].\rootmarking$ if and only if 
		$s\in \hat{\boldT}[pj].\bagset$ and $s\notin \hat{\boldT}[p].\bagset$. 
\end{enumerate}
\end{defi}

Let $\concreteprojection:\concretebags(G,t)\rightarrow \concretebags(t)$ be a function such that 
$\concreteprojection(\bagset,\bagedge,\vertexextension,\neighbourextension,\edgeextension,\rootmarking) = (\bagset,\bagedge)$ 
for each {${(G,t)}$-concrete} bag $(\bagset,\bagedge,\vertexextension,\neighbourextension,\edgeextension,\rootmarking)\in \concretebags(G,t)$.
In other words, $\projection$ transforms a $(G,t)$-concrete bag into a $t$-concrete bag by 
erasing the four last coordinates of the former. If $\hat{\boldT}$ is a term in $\myterms(\concretebags(G,t))$ then we let 
$\concreteprojection(\hat{\boldT})$ be the term in $\myterms(\concretebags(t))$ which is obtained by 
setting $\concreteprojection(\hat{\boldT})[p] = \concreteprojection(\hat{\boldT}[p])$ for each position $p\in \positions(\hat{\boldT})$.

\begin{thm}
\label{theorem:AllConcreteDecompositions}
Let $G$ be a connected graph and let $\boldT\in \myterms(\concretebags(t))$. Then $\boldT$ is a $t$-concrete 
tree decomposition of $G$ if and only if $|V_{\graph(\boldT)}|=|V_G|$ and there exists a $(G,t)$-concrete tree decomposition 
$\hat{\boldT}\in \myterms(\concretebags(G,t))$ such that $\boldT = \concreteprojection(\hat{\boldT})$. 
\end{thm}
\begin{proof}
Assume that $G=\graph(\boldT)$. Then we have $|V_G| = |V_{\graph(\boldT)}|$. We will show how to construct a 
$(G,t)$-concrete tree decomposition $\hat{\boldT}\in \myterms(\concretebags(G,t))$ such that 
$\concreteprojection(\hat{\boldT})= \boldT$. Clearly, we must have 
$\positions(\hat{\boldT}) = \positions(\boldT)$, and for each $p\in \positions(\boldT)$, 
we must have $\hat{\boldT}[p].\bagset = \boldT[p].\bagset$ and 
$\hat{\boldT}[p].\bagedge = \boldT[p].\bagedge$. Additionally, the set $\boldT[p].\rho$ is 
completely determined by the sets $\boldT[p].\bagset$ and $\boldT[p'].\bagset$, where $p'$ is 
the parent of $p$. Therefore, 
it is enough to specify the functions $\hat{\boldT}[p].\vertexextension$, $\hat{\boldT}[p].\neighbourextension$, and 
the set $\hat{\boldT}[p].\edgeextension$ for each $p\in \positions(\boldT)$. 

Let $\beta:\positions(\boldT)\rightarrow \{ \emptyset, \{e_p\}\}$ 
	be a function such that for each $p\in \positions(\boldT)$, $\beta(p) = \emptyset$ if $\boldT[p].\bagedge = \emptyset$ and 
	$\beta(p) = \{e_p\}$ if $\boldT[p].\bagedge \neq \emptyset$. 

\begin{enumerate}
\item For each $p\in \positions(\boldT)$ and each $s\in \hat{\boldT}[p].\bagset$ we 
	set $\hat{\boldT}[p].\vertexextension(s) = v_{s,P(s,p)}$. 
\item For each $p\in \positions(\boldT)$ such that $\boldT[p].\bagedge \neq \emptyset$, we set 
$\hat{\boldT}[p].\edgeextension = \{e_p\}$.
\item For each $p\in \positions(\boldT)$, and each $s\in \hat{\boldT}[p].\bagset$, we 
set 
\begin{equation*}
\hat{\boldT}[p].\neighbourextension(s) = \left\{
\begin{array}{rcl}
\beta(p)\;\dot{\cup}\;\hat{\boldT}[p1].\neighbourextension(s) 
\;\dot{\cup}\; ... \;\dot{\cup}\; \hat{\boldT}[pr].\neighbourextension(s) & & 
\mbox{if $s\in \boldT[p].\bagedge$,} \\
\\
\hat{\boldT}[p1].\neighbourextension(s) 
\;\dot{\cup}\; ... \;\dot{\cup}\; \hat{\boldT}[pr].\neighbourextension(s) & & 
\mbox{otherwise.} \\
\end{array}
\right.
\end{equation*}
where $p1,...,pr$ are the children of $p$. 
\end{enumerate} 

Now one can check by induction on the height of positions that for each $s\in [t]$ and each 
position $p\in \positions(\hat{\boldT})$, the five conditions {\bf C1}-{\bf C5} of Definition 
\ref{definition:GTDecomposition} are satisfied. This implies that $\hat{\boldT}$ is a $(G,t)$-concrete tree decomposition.

For the converse, suppose that $\hat{\boldT}\in \myterms(\concretebags(G,t))$ is a $(G,t)$-concrete tree decomposition 
such that $\boldT = \concreteprojection(\hat{\boldT})$ and $|V_{\graph(\boldT)}| = |V_{G}|$. 
We will show that $\graph(\boldT)$ is isomorphic to $G$. Let $\vertexiso:V_{\graph(\boldT)}\rightarrow V_G$ 
and $\edgeiso:E_{\graph(\boldT)} \rightarrow E_G$ be functions that are defined as follows for each 
vertex $v_{s,P}\in V_{\graph(\boldT)}$ and each edge $e_p\in E_{\graph(\boldT)}$ respectively.  
 
\begin{equation}
\label{equation:IsoVertex}
\vertexiso(v_{s,P}) = \hat{\boldT}[p].\vertexextension(s) \mbox{ if $P=P(s,p)$ for some $p\in P$.} \\
\end{equation}
\begin{equation}
\label{equation:IsoEdge}
\edgeiso(e_p) = e \mbox{ if $\hat{\boldT}[p].\edgeextension  = \{e\}$}. 
\end{equation}

We claim that the pair $\isomorphism = (\vertexiso,\edgeiso)$ is an isomorphism from 
$\graph(\boldT)$ to $G$. First, let $(e_p,v_{s,P}) \in \incidence_{\graph(\boldT)}$. Then, 
by Condition \ref{graphThree} of Definition \ref{definition:GraphConcreteDecomposition}, 
we have that $s\in \boldT[p].\bagedge$. Therefore, by Equation \ref{equation:IsoEdge} and 
by Condition {\bf C2} of Definition \ref{definition:GTDecomposition}, $(\edgeiso(e_p),\hat{\boldT}[p].\vertexextension(s)) \in 
\incidence_G$. Since by Equation \ref{equation:IsoVertex},  $\vertexiso(v_{s,P}) = \hat{\boldT}[p].\vertexextension(s)$, 
we have that $(\edgeiso(e_p),\vertexiso(v_{s,P}))\in \incidence_{G}$. In other words, 
whenever $(e_p,v_{s,P})\in \incidence_{\graph(\boldT)}$, we have that $(\edgeiso(e_p),\vertexiso(v_{s,P}))\in \incidence_{G}$.
This shows that the pair $\isomorphism$ is a morphism from $\graph(\boldT)$ to $G$ in the sense that 
it preserves adjacencies. In order to show that 
$\isomorphism$ is indeed an isomorphism, we need to prove that the functions $\vertexiso$ and $\edgeiso$ are bijections.

Since by assumption we have that $|V_G|=|V_{\graph(\boldT)}|$, to show that 
$\vertexiso$ is a bijection, it is enough to show that for each vertex $v\in V_G$ there is some vertex 
$v_{s,P}\in V_{\graph(\boldT)}$ such that $\morphism(v_{s,P}) = v$. In other words, it is enough to show that 
$\vertexiso$ is surjective. This proceeds as follows. Let $v\in V_{G}$ and $v_{s,P}\in V_{\graph(\boldT)}$ 
be such that $\vertexiso(v_{s,P})= v$. By Equation \ref{equation:IsoVertex}, there is some $p\in P$ such that 
$\boldT[p].\vertexextension(s)=v$. Therefore, since $P=P(s,p)$, by Condition {\bf C1} of Definition 
\ref{definition:GTDecomposition}, $\boldT[p].\vertexextension(s) = v$ for every $p\in P$. By Conditions 
{\bf C3} and {\bf C4} of Definition \ref{definition:GTDecomposition}, for each edge $e$ such that $(e,v)\in \incidence_G$, 
there exists a (unique) $p\in P$ such that $\hat{\boldT}[p].\edgeextension=\{e\}$.  
Now, let $p\in P$, and $e$ be the unique edge such that $\hat{\boldT}[p].\edgeextension=\{e\}$ and assume that 
$\edgeendpoints(e) = \{v,v'\}$. Then by Condition {\bf C2}, there is some $s'$ such that 
$\hat{\boldT}[p].\bagedge = \{s,s'\}$ and $\hat{\boldT}[p].\vertexextension(s') = v'$. Therefore, by Equation 
\ref{equation:IsoVertex}, we have that $\vertexiso(v_{s',P(s',p)}) = v'$. In other words, we have shown that 
if $\isomorphism(v_{s,P}) = v$, then for each neighbour $v'$ of $v$, there is some $p\in P$ and some $s'\in [t]$
such that $\vertexiso(v_{s',P(s',p)}) = v'$. Since $G$ is connected, this implies that for each 
$v\in V_G$ there is some $v_{s,P}$ such that $\vertexiso(v_{s,P}) = v$. 

Now, it remains to show that the function $\edgeiso$ is also a bijection. Let $e$ be an 
edge in $E_G$, and let $v$ be an endpoint of $e$. Then there is some $v_{s,P}\in \graph(\boldT)$
such that $\vertexiso(v_{s,P}) = v$. By the discussion above, this implies that for some $p\in P$, 
$\hat{\boldT}[p].\edgeextension = \{e\}$. Therefore, by Equation \ref{equation:IsoEdge}, we have that 
$\edgeiso(e_p) = e$. Thus we have shown that for each edge $e\in E_G$ there exists some $p\in \positions(\boldT)$
such that $\edgeiso(e_{p}) = e$. In other words, we have shown that $\edgeiso$ is surjective. Now 
we need to show that $\edgeiso$ is injective. Towards this goal, assume that $\edgeiso(e_{p})=\edgeiso(e_{p'})=e$ for 
some $e\in E_G$ and some distinct positions $p,p'\in \positions(\boldT)$ and assume that 
$\edgeendpoints(e) = \{u,v\}$. Then we have that there exists $s_1,s_2\in [t]$ such that
$\boldT[p].b=\{s_1,s_2\}$ and $\boldT[p].\vertexextension(s_1) = v$ and $\boldT[p].\vertexextension(s_2) = u$. 
Analogously, there exists $s_1',s_2'\in [t]$ such that
$\boldT[p'].b=\{s_1',s_2'\}$ and $\boldT[p'].\vertexextension(s_1') = v$ and $\boldT[p].\vertexextension(s_2') = u$. 
Then, from Equation \ref{equation:IsoVertex} and from the fact that $\vertexiso$ is injective,
we have that $s_1=s_1'$, $P(s_1,p) = P(s_1',p')$, $s_2=s_2'$ and $P(s_2,p)=P(s_2',p)$. 
But by Conditions {\bf C3} and {\bf C4} of Definition \ref{definition:GTDecomposition}, there is a unique 
position $p \in P(s_1,p) = P(s_1,p')$ such that $\boldT[p].\edgeextension = \{e\}$. Therefore, $p=p'$. This shows 
that the function $\edgeiso$ is injective.  
\end{proof}

Note that conditions {\bf C1}-{\bf C5} are local in the sense that
they may be verified at each position $p\in \positions(\hat{\boldT})$ by analysing only the 
concrete bags $\hat{\boldT}[p],\hat{\boldT}[p1],...,\hat{\boldT}[pr]$ 
where $p1,...,pr$ are the children of $p$. This allows us to define a tree automaton $\hat{\treeAutomaton}(G,t)$
over $\concretebags(G,t)$ that accepts a term $\hat{\boldT}\in \myterms(\concretebags(G,t))$ if 
and only if $\hat{\boldT}$ is a $(G,t)$-concrete tree decomposition. 

\begin{lem}
\label{lemma:AutomatonGTDecompositions}
For each positive integer $t$ and each graph $G$ of maximum degree $\maximumdegree$, 
one can construct in time $2^{O(\maximumdegree)}\cdot |V_G|^{O(t)}$ a tree automaton $\hat{\treeAutomaton}(G,t)$
over $\concretebags(G,t)$ accepting the following language. 
\begin{equation}
\lang(\hat{\treeAutomaton}(G,t)) = \{\hat{\boldT}\in \myterms(\concretebags(G,t))\;|\; 
\hat{\boldT} \mbox{ is a $(G,t)$-concrete tree decomposition.} \} 
\end{equation}
\end{lem}
\begin{proof}
Let $\hat{B}_1,...,\hat{B}_r, \hat{B}$ be $(G,t)$-concrete bags for $0\leq r \leq 2$, 
where $\hat{B} = (\bagset,\bagedge,\vertexextension,\neighbourextension,\edgeextension)$ 
$\hat{B}_j = (\bagset_j,\bagedge_j,\vertexextension_j,\neighbourextension_j,\edgeextension_j)$ for $j\in [r]$. 
We say that the tuple $(\hat{B}_1,...,\hat{B}_r,\hat{B})$ is $(G,t)$-compatible if the following conditions 
are satisfied for each $s\in [t]$. 

\begin{enumerate}
	\item If $s\in \bagset_j \cap \bagset$ for some $j\in [r]$ then 
		$\vertexextension(s) = \vertexextension_j(s)$. 
	\item If $s\in \bagedge$ then $\edgeextension = \{e\}$ for some 
		edge $e$ such that $(e,\vertexextension(s)) \in \incidence_{G}$. 
	\item $\neighbourextension(s) = \edgeextension \;\dot\cup\; 
		\neighbourextension_1(s)\; \dot\cup\; ... \;\dot\cup\; \neighbourextension_r(s).$
	\item If $s\in \rho$ then $\neighbourextension(s) = \{e\;|\; (e,\vertexextension(v)) \in \incidence_G\}$.
	\item $s\in \rho_j$ if and only if $s\in \bagset_j$ and $s\notin \bagset$.
\end{enumerate}

We define the tree automaton $\treeAutomaton =  (\states,\concretebags(G,t),\finalstates,\transitionsautomaton)$ as 
follows. 
\newcommand{\aconcretebag}{\mathbold{B}}
\begin{equation}
\begin{array}{c}
\states = \{\astate_{\hat{\aconcretebag}}\;|\; \mbox{ $\hat{B}$ is a $(G,t)$-concrete bag.}\} \hspace{1cm} 
\finalstates = \{\astate_{\hat{\aconcretebag}}\in \states \;|\; \hat{\aconcretebag}.\rootmarking = \hat{\aconcretebag}.\bagset\} \\
\\
\transitionsautomaton = 
\{(\astate_{\hat{\aconcretebag}_1},...,\astate_{\hat{\aconcretebag}_k},\hat{\aconcretebag},\astate_{\hat{\aconcretebag}})\;|\;
(\hat{\aconcretebag}_1,...,\hat{\aconcretebag}_r,\hat{\aconcretebag}) \mbox{ is $(G,t)$-compatible.}
\}
\end{array}
\end{equation}
Now, it can be shown by induction on the height of terms that a term 
$\hat{\boldT}\in \myterms(\concretebags(G,t))$ reaches a state $\astate_{\hat{\aconcretebag}}$ if 
and only if $\boldT[\emptystring] = \hat{\aconcretebag}$ and that conditions {\bf C1}-{\bf C5} of 
Definition \ref{definition:GTDecomposition} are satisfied for each $p\in \positions(\hat{\boldT})$ and 
each $s\in [t]$. In particular, this implies that $\hat{\treeAutomaton}(G,t)$ accepts a term 
$\hat{\boldT}\in \myterms(\concretebags(G,t))$ if and only if $\hat{\boldT}$ is a $(G,t)$-concrete 
tree decomposition. 
\end{proof}

The next lemma states that for each positive integers $t$ and $n$, one can efficiently construct a tree automaton 
$\treeAutomaton(t,n)$ which accepts precisely those $t$-concrete tree decompositions which give rise to graphs 
with $n$ vertices. 

\begin{lem}
\label{lemma:GraphsGivenSize}
Let $t$ and $n$ be positive integers with $t\leq n$. One can construct in time $2^{O(t)}\cdot n^3$ a tree automaton 
$\treeAutomaton(t,n)$ over $\concretebags(t)$ accepting the following language. 
\begin{equation*}
\lang(\treeAutomaton(t,n)) = \{\boldT\in \myterms(\concretebags(t))\;|\; |V_{\graph(\boldT)}|=n\} 
\end{equation*}
\end{lem}
\begin{proof}
Let $\boldT$ be a term in $\myterms(\concretebags(t))$. Recall that for each 
position $p\in \positions(\boldT)$, and each $s\in [t]$, the set $P(s,p)$ is the unique $s$-maximal 
subset of $\positions(\boldT)$ that contains position $p$. We say that $(s,p)$ is a $\boldT$-root-pair
if $p$ is the root of $P(s,p)$. We note that the number of vertices of the graph $\graph(\boldT)$
is equal to the number of $\boldT$-root-pairs in $[t]\times \positions(\boldT)$. Therefore, in order 
to construct an automaton that accepts precisely the terms $\boldT\in \positions(\boldT)$ 
that give rise to graphs with $n$ vertices, it is enough to define an automaton that 
accepts precisely those terms $\boldT$ that admit $n$ $\boldT$-root-pairs. 

A {\em root marking} for a set $\bagset\subseteq [t]$ is a set $\rho\subseteq \bagset$. 
The automaton $\treeAutomaton(t,n) = (\states,\concretebags(t),\finalstates,\transitionsautomaton)$
is defined as follows. 
\begin{equation*}
\begin{array}{l}
\states = \{ \astate_{\bagset,\rho,j}\;|\; j\in \{0,...,n\},\; \bagset\subseteq [t], \; \rho 
\mbox{ is a root marking for $\bagset$}\} \hspace{1cm} \finalstates = \{\astate_{\bagset,\rho,n}\}\\
\\
\transitionsautomaton = \{(\astate_{\bagset_1,\rho_1,j_1}, ... ,\astate_{\bagset_r,\rho_r,j_r}, 
(\bagset,\bagedge),\astate_{\bagset,\rho,j})\;|\; 0\leq r \leq 2,\; (\bagset,\bagedge)\in \concretebags(t),\; 
j = |\rho| + \sum_{i=1}^r |\rho_i|,\; \\
\hspace{5.5cm} \mbox{for each $s\in [t]$, if $s\in \rho_i$ for some $i\in [r]$ then $s\notin \bagset$}\}
\end{array}
\end{equation*}

Then it follows by induction on the height of terms that a term $\boldT\in \myterms(\concretebags(t))$
reaches a state $\astate_{\bagset,\rho,j}$ if and only if $\boldT[\emptystring].\bagset = \bagset$, 
$(s,\emptystring)$ is a $\boldT$-root-pair for each $s\in \rho$, and $j$ is the number of $\boldT$-root-pairs 
in $[t]\times \positions(\boldT)$. In particular, the number of $\boldT$-root-pairs is equal to $n$
if and only if $\boldT$ reaches some final state of $\treeAutomaton$. Note that since 
$|\concretebags(t)| = 2^{O(t)}$, and since $r\leq 2$, there are at most $2^{O(t)}\cdot n$ 
states in $\states$ and at most $2^{O(t)}\cdot n^{3}$ transitions in $\transitionsautomaton$. 
Therefore, $\treeAutomaton(t,n)$ can be constructed in time $2^{O(t)}\cdot n^{3}$. 
\end{proof}

The main result of this section (Theorem \ref{theorem:AllDecompositions}), 
 follows by a combination of Theorem \ref{theorem:AllConcreteDecompositions}, Lemma
\ref{lemma:AutomatonGTDecompositions} and Lemma \ref{lemma:GraphsGivenSize}.

\begin{thm}
\label{theorem:AllDecompositions}
Let $G$ be a connected graph of treewidth $t$ and maximum degree $\Delta$. Then one can construct 
in time $2^{O(\Delta\cdot t)}\cdot |V_G|^{O(t)}$ a tree automaton $\automaton(G,t)$ over $\concretebags(t)$ 
such that for each $\boldT\in \myterms(\concretebags(t))$, $\boldT\in \lang(\automaton(G,t))$ if and only 
if $\boldT$ is a concrete tree decomposition of $G$.
\end{thm}
\begin{proof}
By Lemma \ref{lemma:AutomatonGTDecompositions}, one can construct in time $2^{O(\Delta\cdot t)}\cdot |V_G|^{O(t)}$
a tree automaton $\hat{\treeAutomaton}(G,t)$ over $\concretebags(G,t)$ which accepts precisely the 
$(G,t)$ concrete tree decompositions that belong to $\myterms(\concretebags(G,t))$. 

Therefore, the tree automaton $\concreteprojection(\treeAutomaton)$ accepts precisely those $t$-concrete tree 
decompositions $\boldT \in \myterms(\concretebags(t))$ such that $\boldT = \concreteprojection(\hat{\boldT})$ 
for some $(G,t)$-concrete tree decomposition $\hat{\boldT}\in \lang(\hat{\treeAutomaton}(G,t))$. Note that 
$\concreteprojection(\hat{\treeAutomaton}(G,t))$ can be constructed in time $O(|\hat{\treeAutomaton}(G,t)|)$
by Lemma \ref{lemma:Projection}. 

Now, by Lemma \ref{lemma:GraphsGivenSize} we can construct in time $2^{O(t)}\cdot |V_{G}|^{O(1)}$ 
a tree automaton $\treeAutomaton(t,|V_G|)$ over $\concretebags(t)$ which accepts a $t$-concrete tree 
decomposition in $\myterms(\concretebags(t))$ if and only if $|V_{\graph(\boldT)}| = |V_{G}|$. 

Therefore if we set $\treeAutomaton(G,t) = \concreteprojection(\hat{\treeAutomaton}(G,t))\cap \treeAutomaton(t,|V_G|)$, then 
we have that $\treeAutomaton(G,t)$ accepts precisely those $t$-concrete tree decompositions $\boldT\in \concretebags(t)$
such that $|V_{\graph(\boldT)}| = |V_{G}|$ and $\boldT = \concreteprojection(\hat{\boldT})$ for some $(G,t)$-concrete
tree decomposition $\hat{\boldT}$. By Lemma \ref{lemma:EmptynessIntersection}, $\treeAutomaton(G,t)$ can be constructed in time 
$2^{O(\maximumdegree \cdot t)}\cdot |V_G|^{O(t)}$.  
\end{proof}

\section{$(\varphi,t)$-Supergraphs}
\label{section:SupergraphsAndCompletions}

Let $\varphi$ be a $\cmso$ sentence, and $t$ be a positive integer.
Let $\agraph$ and $\agraph'$ be  graphs. 
We say that $\agraph'$ is a $(\varphi,t)$-supergraph of $\agraph$ if the following three conditions are satisfied: 
${\agraph'\models \varphi}$,  $\agraph'$ has treewidth at most $t$, and $\agraph$ is a subgraph of $\agraph'$.

\begin{lem}
\label{lemma:SubdecompositionsSplitAutomatonMSO}
Let $\varphi$ be a $\cmso$ sentence and $t$ be a positive integer. 
Then a  graph $\agraph$ has a $(\varphi,t)$-supergraph 
if and only if there exists a $(t+1)$-concrete tree decomposition 
$\boldT \in \lang(\subdecompositions(\treeAutomaton(\varphi,t+1))$
such that $\graph(\boldT)$ is isomorphic to $\agraph$.
\end{lem}
\begin{proof}
Assume that $\agraph$ is a  graph that has a
$(\varphi,t)$-supergraph $\agraph'$. Then $\agraph'$ satisfies $\varphi$, 
$\agraph'$ has treewidth at most $t$, and $\agraph$ is a subgraph of $\agraph'$. 
By Observation \ref{observation:TreeDecompositionsConcreteTreeDecompositions}, $\agraph'$ has a $(t+1)$-concrete 
tree decomposition $\boldT'\in \myterms(\concretebags(t+1))$, and therefore by Theorem 
\ref{theorem:AutomatonMSO}, $\boldT' \in \lang(\treeAutomaton(\varphi,t))$.
Since $\agraph$ is a subgraph of $\agraph'$, by Theorem \ref{theorem:SubdecompositionLemma},
$\boldT'$ has a sub-decomposition $\boldT$ which is a $(t+1)$-concrete tree decomposition of 
$\agraph$. Therefore, $\boldT$ belongs to $\subdecompositions(\treeAutomaton(\varphi,t+1))$.

For the converse, let $\boldT\in \lang(\subdecompositions(\treeAutomaton(\varphi,t+1)))$ and 
let $\boldT$ be a $(t+1)$-concrete tree decomposition of $\agraph$. 
Then $\boldT$ is a sub-decomposition of some $(t+1)$-concrete tree decomposition $\boldT'$ 
in $\lang(\treeAutomaton(\varphi,t+1))$. By Theorem \ref{theorem:AutomatonMSO}, 
$\boldT'$ is a $(t+1)$-concrete tree decomposition of some graph $\agraph'$ of treewidth at most $t$ such that 
$\agraph'\models \varphi$. Since $\boldT$ is a sub-decomposition of $\boldT'$, 
by Theorem \ref{theorem:SubdecompositionLemma}, $\agraph$ is a subgraph of $\agraph'$. 
Therefore, $\agraph'$ is a $(\varphi,t)$-supergraph of $\agraph$. 
\end{proof}

We note that Lemma \ref{lemma:SubdecompositionsSplitAutomatonMSO} alone does yield 
an algorithm to determine whether a graph $\agraph$ has a $(\varphi,t)$-supergraph. If $\agraph$
does not admit such a supergraph, then no $(t+1)$-concrete tree decomposition $\agraph$ belongs to 
$\lang(\subdecompositions(\treeAutomaton(\varphi,t+1)))$. However, if $\agraph$ does admit 
a $(\varphi,t)$-supergraph, then Theorem \ref{theorem:SubDecompositions} only guarantees 
that some $(t+1)$-concrete tree decomposition $\boldT$ of 
$\agraph$ belongs to $\subdecompositions(\treeAutomaton(\varphi,t+1))$. The problem is 
that $\agraph$ may have infinitely many $(t+1)$-concrete tree
decompositions and we do not know a priori which of these should be tested for membership in $\lang(\subdecompositions(\treeAutomaton(\varphi,t+1)))$. 

The issue described above can be remedied with the results from Section \ref{section:AllDecompositions}.
More specifically, from Theorem \ref{theorem:AllDecompositions} we have that for any given connected graph
$G$ of treewidth $t$ and maximum degree $\Delta$, one can construct a tree automaton $\treeAutomaton(G,t+1)$ over 
$\concretebags(t+1)$ which accepts a $(t+1)$-concrete tree decomposition $\boldT$ if and only if 
the graph $\graph(\boldT)$ is isomorphic to $G$. Note that $\lang(\treeAutomaton(\agraph,t+1))$ is an infinite language that contains $(t+1)$-concrete tree 
decompositions of $G$ of all shapes and sizes.  
Therefore, a connected graph $\agraph$ has a $(\varphi,t)$-supergraph if and only if 

\begin{equation}
\label{equation:Solution}
\lang(\treeAutomaton(\agraph,t+1))\; \cap\; 
\lang(\subdecompositions(\treeAutomaton(\varphi,t+1)))\neq \emptyset.
\end{equation}

The next theorem states that Equation \ref{equation:Solution} yields an efficient algorithm 
for testing whether connected graphs of bounded degree have a $(\varphi,t)$-supergraph.

\begin{thm}[Main Theorem]
\label{theorem:MainTheoremGraphCompletion}
There is a computable function $f$, and an algorithm $\algorithm$ that takes as input 
a $\cmso$ sentence $\varphi$, a positive integer $t$, and 
a connected  graph $\agraph$ of maximum degree $\Delta$,
and determines in time $f(|\varphi|,t)\cdot 2^{O(\Delta\cdot t)}\cdot |\agraph|^{O(t)}$ 
whether $\agraph$ has a $(\varphi,t)$-supergraph. 
\end{thm}
\begin{proof}

By Lemma \ref{lemma:SubdecompositionsSplitAutomatonMSO}, $\agraph$ has a $(\varphi,t)$-supergraph
if and only if there exists some $\boldT\in \lang(\subdecompositions(\treeAutomaton(\varphi,t+1)))$ 
such that $\boldT$ is a $(t+1)$-concrete tree decomposition of $\agraph$. 
By Theorem \ref{theorem:AllDecompositions}, $\lang(\treeAutomaton(\agraph,t+1))$
accepts a $(t+1)$-tree decomposition of $\agraph$ if and only if $\graph(\boldT)$  is isomorphic to 
$G$. Therefore, $\agraph$ has a $(\varphi,t)$-supergraph if and only the 
intersection of $\lang(\treeAutomaton(\agraph,t+1))$ with $\lang(\subdecompositions(\treeAutomaton(\varphi,t+1)))$ is nonempty. 

By Theorem \ref{theorem:AllDecompositions}, the tree-automaton $\treeAutomaton(\agraph,t+1)$ 
can be constructed in time $2^{O(\Delta\cdot t)}\cdot |\agraph|^{O(t)}$, and therefore the size of 
$\treeAutomaton(\agraph,t+1)$ is bounded by $2^{O(\Delta\cdot t)}\cdot |\agraph|^{O(t)}$. 
By Theorem \ref{theorem:SubDecompositions} 
and Theorem \ref{theorem:AutomatonMSO}, the tree-automaton $\subdecompositions(\treeAutomaton(\varphi,t+1))$ 
can be constructed in time $f(|\varphi|,t)$ for some computable function $f:\N^2\rightarrow \N$, and therefore, 
the size of $\subdecompositions(\treeAutomaton(\varphi,t))$ is bounded by $f(|\varphi|,t)$. 

Finally, given tree automata $\treeAutomaton_1$ and $\treeAutomaton_2$, one can determine 
whether $\lang(\treeAutomaton_1)\cap \lang(\treeAutomaton_2)\neq \emptyset$ in time 
$O(|\treeAutomaton_1|\cdot |\treeAutomaton_2|)$ (Lemma \ref{lemma:EmptynessIntersection}). In particular, 
one can determine whether 
$$\lang(\treeAutomaton(\agraph,t+1))\cap \lang(\subdecompositions(\treeAutomaton(\varphi,t+1)))\neq \emptyset$$
in time ${f(|\varphi|,t)\cdot 2^{O(\Delta\cdot t)}\cdot |\agraph|^{O(t)}}$.
\end{proof}

\section{Contraction Closed Graph Parameters}
\label{section:PlanarDiameterImprovement}
\label{section:ContractionClosedParameters}

In this section we deal with simple graphs, i.e., graphs without loops or multiple edges. Therefore, 
we may write $\{u,v\}$ to denote an edge $e$ whose endpoints are $u$ and $v$. Additionally, whenever
speaking of a property specified by an CMSO formula $\varphi$, we assume that $\varphi$ ensures 
that its models are simple graphs. 

Let $\agraph$ be a graph and $\{u,v\}$ be an edge of $\agraph$. We let $\agraph/uv$
denote the graph that is obtained from $\agraph$ by deleting the edge $\{u,v\}$ and by merging vertices 
$u$ and $v$ into a single vertex $x_{uv}$. We say that $\agraph/uv$ is obtained from $\agraph$ by an edge-contraction. 
We say that a graph $\agraph'$ is a {\em contraction} of $\agraph$ if 
$\agraph'$ is obtained from $\agraph$ by a sequence of edge contractions. 
We say that $\agraph'$ is a minor of $\agraph$ if $\agraph'$ is a contraction of 
some subgraph of $\agraph$. We say that a graph $\agraph$ is an apex graph if after deleting some 
vertex of $\agraph$ the resulting graph is planar. 

\newcommand{\graphproperty}{\mathcal{P}}
A graph parameter is a function $\graphparameter$ mapping graphs to non-negative integers in such a way that 
$\graphparameter(\agraph) = \graphparameter(\agraph')$ whenever $\agraph$ is isomorphic to $\agraph'$. We say that $\graphparameter$ 
is contraction closed if $\graphparameter(\agraph')\leq \graphparameter(G)$ whenever $\agraph'$ is a contraction of 
$\agraph$. 

A graph property is simply a set $\graphproperty$ of graphs. We say that a property 
$\graphproperty$ is contraction-closed if for every two graphs $\agraph, \agraph'$ for which $\agraph'$ is a 
contraction of $\agraph$, the fact that $\agraph\in \graphproperty$ implies that $\agraph'\in \graphproperty$. 

\subsection{Diameter Improvement Problems}
\label{subsection:DiameterImprovement}

Let $u$ and $v$ be vertices in an  graph $\agraph$. The distance from $u$
to $v$, denoted by $\mathit{dist}(u,v)$ is the number of edges in the shortest path from $u$ to $v$.
If no such path exists, we set $\mathit{dist}(u,v)=\infty$. The diameter of $\agraph$ is defined as 
$\mathit{diam}(\agraph) = \max_{u,v}\mathit{dist}(u,v)$. In the {\sc planar diameter improvement} problem (PDI),
we are given an  graph $\agraph$ and a positive integer $d$, and the goal 
is to determine whether $\agraph$ has a planar supergraph $\agraph'$ of diameter at most $d$. As mentioned 
in the introduction, there is an algorithm that solves the PDI problem in time $f(d)\cdot |\agraph|^{O(1)}$, where 
$f:\N\rightarrow \N$ is not known to be computable. Additionally, even the problem of determining whether PDI
admits an algorithm running in time $f_1(d)\cdot |G|^{f_2(d)}$ for computable functions $f_1,f_2$ remains open 
for more than two decades \cite{FellowsLangston1989,CohenGoncalvesKimPaulSauThilikosWeller15}. The next 
theorem solves this problem when the input graphs are connected and have bounded degree. 

\begin{thm}
\label{theorem:DiameterImprovementPlanar}
There is a computable function $f:\N\rightarrow \N$, and an algorithm $\algorithm$ that 
takes as input, a positive integer $d$, and a connected  graph $\agraph$
of maximum degree $\maximumdegree$, and determines in time $f(d)\cdot 2^{O(\maximumdegree\cdot d)}\cdot |\agraph|^{O(d)}$ 
whether $\agraph$ has a planar supergraph $\agraph'$ of diameter at most $d$. 
\end{thm}
\begin{proof}
It should be clear that there is an algorithm that takes a positive integer $d$ as input, and constructs, 
in time $O(d)$, a $\cmso$ formula $\mathit{Diam}_d$ which is true on a  graph $\agraph'$ if and only if $\agraph'$ has 
diameter at most $d$. Additionally, using Kuratowski's theorem, and the fact that minor relation is CMSO expressible, 
one can define a CMSO formula $\mathit{Planar}$ which is true on a  graph $\agraph'$ if and only if $\agraph'$ 
	is planar. Finally, it can be shown that any planar graph of diameter at most $d$ has treewidth $O(d)$ (see Lemma 1 of \cite{Eppstein2000}). 
Therefore, by setting $\varphi=\mathit{Diam}_d \wedge \mathit{Planar}$, $t=O(d)$, and by renaming $f(|\varphi|,t)$ to $f(d)$ in 
Theorem \ref{theorem:MainTheoremGraphCompletion}, we have an algorithm running in time 
$f(d)\cdot 2^{O(\maximumdegree\cdot d)}\cdot |\agraph|^{O(d)}$ to determine whether $\agraph$ has a planar supergraph 
$\agraph'$ of diameter at most $d$. 
\end{proof}

We note that the algorithm $\algorithm$ of Theorem \ref{theorem:DiameterImprovementPlanar} does not impose 
any restriction on the degree of a prospective supergraph $\agraph'$ of $\agraph$. Theorem \ref{theorem:DiameterImprovementPlanar}
can be generalized to the setting of graphs of constant genus as follows.  

\begin{thm}
\label{theorem:DiameterImprovementGenus}
There is a computable function $f:\N\times \N\rightarrow \N$, and an algorithm $\algorithm$ that 
takes as input, positive integers $d,g$, and a connected  graph $\agraph$
of maximum degree $\maximumdegree$, and determines in time 
$f(d,g)\cdot 2^{O(\maximumdegree\cdot d)}\cdot |\agraph|^{O(d\cdot g)}$ whether 
$G$ has a supergraph $\agraph'$ of genus at most $g$ and diameter at most $d$. 
\end{thm}
\begin{proof}
It can be shown that there is an explicit algorithm that takes a positive integer $g$ as input and constructs 
a CMSO sentence $\mathit{Genus}_g$ that is true on a graph $\agraph'$ if and only if $\agraph'$ has 
genus at most $g$ \cite{AdlerGroheKreutzer2008}. Additionally, 
it can be shown that graphs of genus $g$ and diameter $d$ have treewidth at most $O(g\cdot d)$ \cite{Eppstein2000}. 
Therefore, by setting $\varphi = \mathit{Genus}_g\wedge \mathit{Diam}_d$, $t=O(g\cdot d)$, and 
by renaming $f(|\varphi|,t)$ to $f(g,d)$ in Theorem \ref{theorem:MainTheoremGraphCompletion}, 
we have that one can determine in time $f(d,g)\cdot 2^{O(\maximumdegree\cdot d)}\cdot |\agraph|^{O(d\cdot g)}$
whether $\agraph$ has a supergraph $\agraph'$ of genus at most $g$ and diameter at most $d$. 
\end{proof}

A graph is $1$-outerplanar if it can be embedded in the plane in such a way that every vertex lies in the 
outer face of the embedding. A graph is $k$-outerplanar if it can be embedded in the plane in such a way that after 
deleting all vertices in the outer face, the remaining graph is ${(k-1)}$-outerplanar. In \cite{CohenGoncalvesKimPaulSauThilikosWeller15}
Cohen et al. have considered the {\sc $k$-outerplanar diameter improvement} problem {($k$-OPDI)}, a variant of 
the PDI problem in which the target supergraph is required to be $k$-outerplanar instead of planar. In particular, 
they have shown that the $1$-OPDI problem can be solved in polynomial time. The complexity of the $k$-OPDI problem 
with respect to explicit algorithms was left as an open problem for $k\geq 2$. The next theorem states 
that for each fixed $k$, $k$-OPDI is strongly uniformly fixed parameter tractable with respect to the parameter $d$ 
on connected graphs of bounded degree. 

\begin{thm}
\label{theorem:DiameterImprovementOuterplanar}
There is a computable function $f:\N\times \N\rightarrow \N$, and an algorithm $\algorithm$ that 
takes as input, positive integers $d,k$, and a connected  graph $\agraph$
of maximum degree $\maximumdegree$, and determines in time $f(k,d)\cdot 2^{O(\maximumdegree\cdot k)}\cdot |\agraph|^{O(k)}$ whether 
$G$ has a $k$-outerplanar supergraph $\agraph'$ of diameter at most $d$. 
\end{thm}
\begin{proof}
There is an algorithm that takes an integer $k$ as input and constructs in time $O(k)$ 
a CMSO sentence $\mathit{Outer}_k$ that is true on a graph $\agraph$ if and only 
if $\agraph$ is $k$-outerplanar \cite{JaffkeBodlaender2015}. Additionally, it can be shown that 
$k$-outerplanar graphs have treewidth $O(k)$. Therefore, by setting $\varphi = \mathit{Outer}_k\wedge \mathit{Diam}_d$, 
and $t=O(k)$ in Theorem \ref{theorem:MainTheoremGraphCompletion}, it follows that the problem of determining 
whether $\agraph$ has a $k$-outerplanar supergraph of diameter at most $d$ can be decided in time 
$f(k,d)\cdot 2^{O(\maximumdegree\cdot k)}\cdot |\agraph|^{O(k)}$ for some computable function $f:\N\times \N\rightarrow \N$.
\end{proof}

Finally, the {\sc series-parallel diameter improvement} problem (SPDI) consists in determining whether 
a graph $\agraph$ has a series parallel supergraph of diameter at most $d$. The parameterized complexity 
of this problem was left as an open problem in \cite{CohenGoncalvesKimPaulSauThilikosWeller15}. The next 
theorem states that SPDI is strongly uniformly fixed parameter tractable with respect to the parameter
$d$ on connected graphs of bounded degree.  

\begin{thm}
\label{theorem:DiameterImprovementSeriesParallel}
There is a computable function $f:\N \rightarrow \N$, and an algorithm $\algorithm$ that 
takes as input, a positive integer $d$ and a connected  graph $\agraph$ 
of maximum degree $\maximumdegree$, and determines in time $f(d)\cdot 2^{O(\maximumdegree)}\cdot |\agraph|^{O(1)}$ whether 
$\agraph$ has a series-parallel supergraph $\agraph'$ of diameter at most $d$. 
\end{thm}
\begin{proof}
There is a CMSO formula $\mathit{SP}$ which is true on a graph $\agraph'$ if and only if $\agraph'$ is series parallel. 
Additionally, series parallel graphs have treewidth at most $2$. Therefore, by setting 
$\varphi=\mathit{SP}\wedge \mathit{Diam}_d$ and $t=O(1)$ in Theorem \ref{theorem:MainTheoremGraphCompletion}, it follows 
that the problem of determining whether $\agraph$ has a series-parallel supergraph of diameter at most $d$ can be 
decided in time $f(d)\cdot 2^{O(\maximumdegree)} \cdot |\agraph|^{O(1)}$ for some computable function $f$. 
\end{proof}

\subsection{Contraction Bidimensional Parameters}

\newcommand{\gammagraph}{\Gamma}

Fomin, Golovach and Thilikos \cite{FominGolovachThilikos2011Contraction} have defined a sequence 
$\{\gammagraph_k\}_{k\in \N}$ of graphs and have shown that these graphs serve as 
obstructions for small treewidth on $H$-minor free graphs, whenever $H$ is an apex graph. 
More precisely, they have proved the following result. 

\begin{thm}[Fomin-Golovach-Thilikos \cite{FominGolovachThilikos2011Contraction}]
\label{theorem:FominGolovachThilikosGamma}
For every apex graph $H$, there is a $c_H>0$ such that every connected $H$-minor-free 
graph of treewidth at least $c_H\cdot k$ contains $\gammagraph_k$ as a contraction. 
\end{thm}

We say that a graph parameter $\graphparameter$ is Gamma-unbounded if there is a computable function 
$\alpha:\N\rightarrow \N$ such that $\alpha\in \omega(1)$, and $\graphparameter(\gammagraph_k) \geq \alpha(k)$ for
every $k\in \N$. 

We say that a parameter $\graphparameter$ is effectively CMSO definable if there is a computable function $f:\N\rightarrow \N$, 
and an algorithm $\algorithm$ that takes as input a positive integer $k$ and constructs, in time at most $f(k)$,
a $\cmso$-sentence $\varphi$ which is true on an  graph 
$\agraph$ if and only if $\graphparameter(\agraph)\leq k$. The following theorem is a corollary of 
Theorem \ref{theorem:MainTheoremGraphCompletion} and Theorem \ref{theorem:FominGolovachThilikosGamma}. 

\begin{thm}
\label{theorem:MostGeneral}
Let $\graphparameter$ be a Gamma-unbounded effectively CMSO definable graph parameter, and let $\mathcal{P}$ be a 
CMSO definable graph property excluding some apex graph $H$ as a minor. Then there is a computable function 
$f:\N \rightarrow \N$ and an algorithm $\algorithm$ that takes as input a positive integer $k$, and a connected  
 graph $\agraph$ of maximum degree $\Delta$, and determines, in time 
$f(k)\cdot 2^{O(\Delta\cdot f(k))}\cdot |\agraph|^{f(k)}$, whether $\agraph$ has a supergraph $\agraph'$ 
such that $\agraph'\in \graphproperty$ and $\graphparameter(\agraph')\leq k$.
\end{thm}

Note that similarly to the case of diameter improvement problem, if $\graphparameter$ is an unbounded effectively 
CMSO definable graph parameter, then we can determine whether a graph $\agraph$ has an $r$-outerplanar 
supergraph $\agraph'$ with $\graphparameter(\agraph')\leq k$ in time $f(r,k)\cdot 2^{O(\Delta\cdot r)}\cdot |G|^{O(r)}$
for some computable function $f:\N\times \N\rightarrow \N$. In other words, this problem, for connected bounded degree graphs, 
is strongly uniformly fixed parameter tractable with respect to the parameter $\graphparameter$ for each fixed $r$. 

\begin{defi}
\label{definition:ContractionBidimensionalParameter}
A graph parameter $\graphparameter$ is contraction-bidimensional if the following conditions are
satisfied. 
\begin{enumerate}
\setlength\itemsep{3pt}
\item $\graphparameter$ is contraction-closed. 
\item If $\agraph$ is a graph which has $\gammagraph_k$ as a contraction, 
	then $\graphparameter(\agraph)\geq \Omega(k^2)$. 
\end{enumerate}
\end{defi}

For instance, the following parameters are contraction bidimensional.

\begin{enumerate} 
	\item Size of a vertex cover.
	\item Size of a feedback vertex set.
	\item Size of a minimum maximal matching.
	\item Size of a dominating set.
	\item Size of a edge dominating set.
	\item Size of a clique traversal set. 
\end{enumerate}

\begin{thmC}[\cite{FominGolovachThilikos2011Contraction,FominLokshtanovSaurabhThilikos2010}]
\label{theorem:BidimensionalityTreewidth}
Let $\graphparameter$ be a bidimensional parameter. Then if $\graphparameter(\agraph)\leq k$, 
the treewidth of $\graphparameter$ is at most $O(\sqrt{k})$. 
\end{thmC}

\begin{thm}
\label{theorem:ContractionBidimensionalPlanarImprovement}
For each effectively $CMSO$-definable contraction-bidimensional parameter $\graphparameter$,
there exists a computable function $f:\N\rightarrow \N$ and an algorithm $\algorithm$ that takes as 
input a positive integer $k$, and a connected  graph $\agraph$
of maximum degree $\maximumdegree$, and determines in time  $f(k)\cdot 2^{O(\maximumdegree\cdot \sqrt{k})}\cdot |\agraph|^{O(\sqrt{k})}$ 
whether $\agraph$ has a planar supergraph $\agraph'$ with $\graphparameter(\agraph')\leq k$.
\end{thm}
\begin{proof}
Since $\graphparameter$ is effectively CMSO definable, there is some computable function $\alpha$ and an algorithm 
that take a positive integer $k$ as input and constructs in time $\alpha(k)$ a CMSO sentence $\varphi_{\graphparameter}$
which is true on an  graph $\agraph$ if and only if $\graphparameter(\agraph)\leq k$. Additionally, 
by Theorem \ref{theorem:BidimensionalityTreewidth}, if $\graphparameter(\agraph)\leq k$, then the treewidth 
of $\agraph$ is bounded by $\sqrt{k}$. Therefore, by applying Theorem \ref{theorem:MainTheoremGraphCompletion} 
with $\varphi = \varphi_{\graphparameter}$, and $t=O(\sqrt{k})$, the theorem follows. 
\end{proof}

For instance, Theorem \ref{theorem:ContractionBidimensionalPlanarImprovement} states that for some computable 
function $f:\N\rightarrow \N$, one can determine in time 
$f(k)\cdot 2^{O(\maximumdegree\cdot \sqrt{k})}\cdot |\agraph|^{O(\sqrt{k})}$ whether $\agraph$ has a
planar supergraph $\agraph'$ with feedback vertex set at most $k$. We note that 
in view of Theorem \ref{theorem:BidimensionalityTreewidth}, the planarity requirement of Theorem 
\ref{theorem:ContractionBidimensionalPlanarImprovement} can be replaced for any CMSO definable property $\graphproperty$
which excludes some apex graph as a minor.

\section*{Acknowledgements}
This work was supported by the Bergen Research Foundation and by the Research Council of Norway (Proj. no. 288761). The author thanks 
Michael Fellows, Fedor Fomin, Petr Golovach, Daniel Lokshtanov and Saket Saurabh 
for interesting discussions. The author also thanks anonymous reviewers for 
several useful comments and suggestions for improvement.

\bibliographystyle{alpha}
\bibliography{supergraphsSatisfyingCMSOProperties}

\appendix

\section{Proof of Theorem \ref{theorem:AutomatonMSO}} 
\label{ProofTheoremAutomatonMSO}

\newcommand{\boldS}{{\mathbf{S}}}
\newcommand{\newslicealphabet}{\mathbold{\Sigma}}
\newcommand{\vertexlabel}{\Gamma_1}
\newcommand{\edgelabel}{\Gamma_2}
\newcommand{\vertexlabeling}{\rho}
\newcommand{\edgelabeling}{\xi}
\newcommand{\interpretedAlphabet}{\newslicealphabet(c,\vertexlabel,\edgelabel,\mathcal{X})}
\newcommand{\interpretedAlphabetMinusX}{\newslicealphabet(c,\vertexlabel,\edgelabel,\mathcal{X}\backslash\{X\})}
\newcommand{\mso}{{$\mbox{MSO}$\;}}
\newcommand{\card}{\mathit{card}}
\newcommand{\interpretationbag}{I}
\newcommand{\interpretationdecomposition}{\mathcal{I}}
\newcommand{\interpretationgraphs}{\mathcal{J}}
\newcommand{\interpretationinduced}{\hat{\mathcal{I}}}
\newcommand{\image}{\mathrm{Im}}

\begin{retheorem}[\ref{theorem:AutomatonMSO}]
\label{retheorem:AutomatonMSO}
There exists a computable function $f:\N\times \N \rightarrow \N$ such that for each $\cmso$ sentence
$\varphi$, and each $t\in \N$, one can construct in time $f(|\varphi|,t)$ a tree-automaton 
$\automaton(\varphi,t)$ accepting the following tree language. 
\begin{equation}
\lang(\treeAutomaton(\varphi,t))  = \{\boldT\in \myterms(\concretebags(t))\;|\; \graph(\boldT) \models \varphi\}.
\end{equation}
\end{retheorem}
\begin{proof}

Let $S$ and $S'$ be sets and let $R\subseteq S\times S'$ be a relation. For each element 
$s\in S$, we let $\image(R,s) = \{s'\;:\; (s,s')\in R\}$ be the image of $s$ under $R$.

Let $\mathcal{X}$ be a set of first-order variables and second-order variables
and let $\agraph$ be a graph. An interpretation of $\mathcal{X}$ in $\agraph$ is a function 
$\interpretationgraphs:\mathcal{X}\rightarrow (V\cup E) \cup (2^{V}\cup 2^{E})$ that assigns a vertex 
$\interpretationgraphs(x)$ to each first-order vertex-variable $x$, an edge $\interpretationgraphs(y)$ to
each first-order edge-variable $y$, a set of vertices $\interpretationgraphs(X)$ 
to each second-order vertex-variable $X$, and a set of edges $\interpretationgraphs(Y)$ to each 
second-order edge-variable $Y$. The semantics of a formula $\varphi$ with free variables $\mathcal{X}$ being 
true on a graph $G$ under interpretation $\interpretationgraphs$ is the standard one. 

Let $\mathcal{X}$ be a set of free first-order and
second-order variables, and let $(\bagset,\bagedge)$ be a $t$-concrete bag.
An interpretation of $\mathcal{X}$ in $(\bagset,\bagedge)$ 
is a relation $\interpretationbag \subseteq \mathcal{X}\times (\bagset \cup \{\bagedge\})$ 
such that the following conditions are satisfied: $\image(\interpretationbag,x) \subseteq \bagset$ and $|\image(\interpretationbag,x)|\leq 1$ for each
first-order vertex-variable $x$; $\image(\interpretationbag,X)\subseteq \bagset$ for each second-order vertex-variable $X$; 
$\image(\interpretationbag,y)\subseteq \{\bagedge\}$ for each first-order edge-variable $y$; and 
$\image(\interpretationbag,Y)\subseteq \{\bagedge\}$ for each second-order variable $Y$. 
We let $\concretebags(t,\mathcal{X})$ be the set of all triples of the form $(\bagset,\bagedge,\interpretationbag)$ where 
$(\bagset,\bagedge)$ is a $t$-concrete bag and $\interpretationbag$ is an interpretation of $\mathcal{X}$ in $(\bagset,\bagedge)$. 

If $\boldT$ is a $t$-concrete decomposition in $\myterms(\concretebags(t))$, then an interpretation of
$\mathcal{X}$ in $\boldT$ is a function $\interpretationdecomposition:\positions(\boldT)\rightarrow \concretebags(t,\mathcal{X})$
where for each position $p\in \positions(\boldT)$, $\interpretationdecomposition(p)$ is an interpretation of 
$\mathcal{X}$ in the $t$-concrete bag $\boldT[p]$, and for each first-order vertex-variable $x$ (edge-variable $y$), 
there is at most one position $p\in \positions(\boldT)$ such that $|\image(\interpretationbag,x)|=1$ ($|\image(\interpretationbag,y)|=1$).
If $\interpretationdecomposition$ is an interpretation of $\mathcal{X}$ in $\boldT$, then the
interpretation of $\mathcal{X}$ in $\graph(\boldT)$ induced by $\interpretationdecomposition$ is
the function 
$$\interpretationinduced:\mathcal{X}\rightarrow (V_{\graph(\boldT)}\cup E_{\graph(\boldT)}) \cup (2^{V_{\graph(\boldT)}} \cup 2^{E_{\graph(\boldT)}})$$
defined as follows.  
\begin{enumerate}
	\item For each $s\in [t]$, each $s$-maximal component $P\subseteq \positions(\boldT)$, and each 
		first-order or second-order vertex-variable $\mathbf{x}$, the vertex 
		$v_{s,P}$ belongs to $\image(\interpretationinduced,\mathbf{x})$ if and only if there exists
		a position $p\in P$ such that $(\mathbf{x},s)\in \interpretationdecomposition(p)$. 
	\item For each $p\in \positions(\boldT)$ such that $\boldT[p].\bagedge\neq \emptyset$, and each first-order or second-order edge-variable
		$\mathbf{y}$, the edge $e_{p}$ belongs to $\image(\interpretationinduced,\mathbf{y})$ if and only if 
		$(\mathbf{y},\boldT[p].\bagedge) \in \interpretationdecomposition(p)$. 
\end{enumerate}

If $\boldT$ is a $t$-concrete decomposition in $\myterms(\concretebags(t))$ and $\interpretationdecomposition$ is an 
interpretation of $\mathcal{X}$ in $\boldT$, then we write $\boldT^{\mathcal{I}}$ to denote the term in 
$\myterms(\concretebags(t,\mathcal{X}))$ where $\boldT^{\mathcal{I}}[p] = (\boldT[p],\mathcal{I}(p))$ for each position $p\in \positions(\boldT)$.
We say that $\boldT^{\mathcal{I}}$ is an interpreted term. We note that one can straightforwardly construct a 
tree automaton $\treeAutomaton(t,\mathcal{X})$ over the alphabet $\concretebags(t,\mathcal{X})$ that accepts precisely the
interpreted terms in $\myterms(\concretebags(t,\mathcal{X}))$.  

For each \msotwo formula $\psi$ with free variables $\mathcal{X}$ 
we will construct a tree-automaton $\treeAutomaton(\psi,t)$ 
over the alphabet $\concretebags(t,\mathcal{X})$ whose language
$\lang(\treeAutomaton(\psi,t))$ consists of all interpreted terms 
$\boldT^{\interpretationdecomposition}\in \myterms(\concretebags(t,\mathcal{X}))$
such that $\graph(\boldT) \models \psi$ under the interpretation $\interpretationinduced$ of 
$\mathcal{X}$ in $\graph(\boldT)$ induced by~$\interpretationdecomposition$. 
The tree-automaton $\treeAutomaton(\psi,t)$ is constructed by induction on the structure of the formula~$\psi$. 

\paragraph{Base Case}
In the base case, the formula $\psi$ is an atomic predicate. There are five cases to be considered. 
Below, we describe the behavior of the tree-automaton $\treeAutomaton(\psi,t)$ in each of these five 
cases. The proper specification of the set of states and set of transitions of each of the tree-automata 
described below is straightforward.
\begin{enumerate}
	\item  If $\psi \equiv (\mathbf{z}_1 = \mathbf{z}_2)$ where $\mathbf{z}_1$ and $\mathbf{z}_2$ are both vertex-variables, 
		both edge-variables, both vertex-set variables or both edge-set variables, then $\treeAutomaton(\psi,t)$ accepts
		a term $\boldT^{\interpretationdecomposition}$ if and only if $\boldT^{\interpretationdecomposition}$ is an 
		interpreted term, and for each position $p\in \positions(\boldT)$, and each $s\in \boldT[p].\bagset$, 
		$(\mathbf{z}_1,s)\in \interpretationdecomposition(p) \Leftrightarrow (\mathbf{z}_2,s)\in \interpretationdecomposition(p)$.
	\item If $\psi \equiv z\in Z$ where either $z$ is a vertex-variable and $Z$ is a vertex-set variable, or $z$ is an edge-variable
		and $Z$ is an edge-set variable, then $\treeAutomaton(\psi,t)$ accepts a term $\boldT^{\interpretationdecomposition}$ if and only if 
		$\boldT^{\interpretationdecomposition}$ is an interpreted term and for each position $p\in \positions(\boldT)$, 
		and each element $r\in \boldT[p].\bagset \cup \{\boldT[p].\bagedge\}$, $(z,r)\in \interpretationdecomposition(p)\Leftrightarrow 
		(Z,r)\in \interpretationdecomposition(p)$.
	\item If $\psi \equiv \incidence(y,x)$ where $y$ is an edge variable and $x$ is a vertex variable, then the automaton 
		$\treeAutomaton(\psi,t)$ accepts a term $\boldT^{\interpretationdecomposition}$ if and only if
		 $\boldT^{\interpretationdecomposition}$ is an interpreted term and there 
		is some position $p\in \positions(\boldT)$ and some $s\in \boldT[p].\bagset$ such that $(x,s)\in \interpretationdecomposition(p)$,
		$s\in \boldT[p].b$, and $(y,\boldT[p],b)\in \interpretationdecomposition(p)$.
	\item If $\psi \equiv \mathit{card}_{a,r}(Z)$ where $0\leq a < r$, $r\geq 2$, and $Z$ is an edge-set variable,
		then the tree automaton $\treeAutomaton(\psi,t)$ accepts $\boldT^{\interpretationdecomposition}$ if 
		and only if $\boldT^{\interpretationdecomposition}$ is an interpreted term, and the number of positions 
		$p\in \positions(\boldT)$ such that $\boldT[p].b\neq \emptyset$ and $(Z,\boldT[p].b)\in \interpretationdecomposition(p)$ 
		is equal to $a \mod r$. 
	\item If $\psi \equiv \mathit{card}_{a,r}(Z)$ where $0\leq a < r$, $r\geq 2$, and $Z$ is a vertex-set variable 
		then the tree automaton $\treeAutomaton(\psi,t)$ accepts $\boldT^{\interpretationdecomposition}$ if and only if 
		$\boldT^{\interpretationdecomposition}$ is an interpreted term and the number of pairs of the form $(s,p)\in [t]\times \positions(\boldT)$
		such that $s\in \boldT[p].B$ and $p$ is the root of $P_{s,p}$ is equal to $a \mod r$. 
\end{enumerate}

\paragraph{Disjunction, conjunction and negation}
The three boolean operations $\vee,\wedge,\neg$ are dealt with using the fact that tree-automata are effectively 
closed under union, intersection and complement (Lemma \ref{lemma:PropertiesOfTreeAutomata}).
Below, we let $\treeAutomaton(t,\mathcal{X})$ be the tree automaton generating the set of interpreted terms over 
$\concretebags(t,\mathcal{X})$. 

\begin{equation}
\begin{array}{c}
\treeAutomaton(\psi \vee \psi',t) = \treeAutomaton(\psi ,t) \cup \treeAutomaton(\psi',t) \\
\\
\treeAutomaton(\psi \wedge \psi',t) = \treeAutomaton(\psi ,t) \cap \treeAutomaton(\psi', t) \\
\\
\treeAutomaton(\neg \psi,t) = \overline{\treeAutomaton(\psi ,t)} \cap 
	\treeAutomaton(t,\mathcal{X})\\
\end{array}
\end{equation}

Observe that in the definition of $\treeAutomaton(\neg \psi,t)$, the 
intersection with the tree-automaton $\treeAutomaton(t,\mathcal{X})$ guarantees 
that all terms in $\lang(\treeAutomaton(\neg \psi,t))$  are interpreted. 

\paragraph{Existential Quantification}
Let $\interpretationbag$ be an interpretation of $\mathcal{X}$ in $(\bagset,\bagedge)$, and 
let $Z$ be either a first-order or a second-order variable in $\mathcal{X}$. We let 
$\interpretationbag - Z = \interpretationbag \cap \left[ (\mathcal{X}\backslash Z) \times (\bagset\cup \{\bagedge\})\right]$ 
be the relation obtained from $\interpretationbag$ by deleting all pairs of the form $(Z,r)$ for $r\in \bagset\cup \{\bagedge\}$. 
To eliminate existential quantifiers we proceed as follows: 
For each variable $Z\in \mathcal{X}$, we let $\mathit{Proj}_{Z}:\concretebags(t,\mathcal{X}) \rightarrow \concretebags(t,\mathcal{X}\backslash Z)$ 
be the map that sends each interpreted $t$-concrete bag $(\bagset,\bagedge,\interpretationbag) \in \concretebags(t,\mathcal{X})$
to the interpreted $t$-concrete bag $(\bagset,\bagedge, \interpretationbag- Z) \in \concretebags(t,\mathcal{X}\backslash Z)$. 
Subsequently, we extend $\mathit{Proj}_{Z}$ homomorphically to terms by setting $\mathit{Proj}_Z(\boldT)[p] = \mathit{Proj}_{Z}(\boldT[p])$ 
for each position $p$ in $\positions(\boldT)$. Finally, we extend $\mathit{Proj}_Z$ to tree languages over $\concretebags(t,\mathcal{X})$
by applying  this map to each term of the language. 
Then we set $$\treeAutomaton(\exists Z \psi(Z),t) = \mathit{Proj}_{Z}(\treeAutomaton(\psi(Z), t)).$$

We note that if $\psi$ is a sentence, i.e., a formula without free variables, then by the end of this 
inductive process all variables occurring in $\psi$ will have been eliminated. 
In this way, the language $\lang(\treeAutomaton(\psi,t))$ will 
consist precisely of the interpreted decompositions $\boldT^{\varepsilon}\in \myterms(\concretebags(t,\emptyset))$ 
where $\varepsilon$ is the empty interpretation, that assigns $\varepsilon(p)= \emptyset$ for each $p\in \positions(\boldT)$,
and $\boldT \models \psi$. Now consider the map that $\mathfrak{m}:\concretebags(t,\emptyset) \rightarrow \concretebags(t)$
that sends each triple $(B,b,\emptyset)\in \concretebags(t,\emptyset)$ to the $t$-concrete bag $(B,b)$. Then, 
by setting $\treeAutomaton(\psi,t) \leftarrow \mathfrak{m}(\treeAutomaton(\psi,t))$, we have that 
$\treeAutomaton(\psi,t)$ accepts a term $\boldT\in \concretebags(t)$ if and only if $\boldT$ is a $t$-concrete tree decomposition
such that $\boldT\models \varphi$. 
\end{proof}

\end{document}